\documentclass[12pt, draftclsnofoot, onecolumn ]{IEEEtran}

\usepackage[dvips]{color}
\usepackage{epsf}
\usepackage{times}
\usepackage{epsfig}
\usepackage{graphicx}

\usepackage{amsmath}
\usepackage{amssymb}
\usepackage{amsxtra}
\usepackage{amsthm}
\usepackage{bbm}

\usepackage{here}
\usepackage{rawfonts}
\usepackage{times}
\usepackage{url}
\usepackage{cite}
\usepackage{comment}
\usepackage[utf8]{inputenc}
\usepackage{caption}
\usepackage{subcaption}

\usepackage{pstricks}
\usepackage{algorithm}
\usepackage{algpseudocode}
\usepackage{lipsum}

\DeclareMathOperator{\tr}{tr}

\newtheorem{theorem}{\bf Theorem}

\newcommand*\dif{\mathop{}\!\mathrm{d}}

\makeatletter

\begin{document}

\title{Millimeter Wave Communications with an Intelligent Reflector: Performance Optimization and Distributional Reinforcement Learning}
\author{ 
	\IEEEauthorblockN{Qianqian Zhang$^1$,  Walid Saad$^1$, and Mehdi Bennis$^2$}
	
	\IEEEauthorblockA{\small
		$^1$Bradley Department of Electrical and Computer Engineering, Virginia Tech, VA, USA, Emails: \url{{qqz93,walids}@vt.edu}. \\
		$^2$Center for Wireless Communications, University of Oulu, Finland, Email: \url{mehdi.bennis@oulu.fi}.  	
	}
}
\maketitle

\vspace{-1.5cm}

\begin{abstract}

In this paper, a novel framework is proposed to optimize the downlink multi-user communication of a millimeter wave base station, which is assisted by a reconfigurable intelligent reflector (IR).  
In particular, a channel estimation approach is developed to measure the channel state information (CSI) in real-time.  
First, for a perfect CSI scenario, the precoding transmission of the BS and the reflection coefficient of the IR are jointly optimized, via an iterative approach,  so as to maximize the sum of downlink rates towards multiple users. 
Next, in the imperfect CSI scenario,  a distributional reinforcement learning (DRL) approach is proposed to learn the optimal IR reflection and maximize the expectation of downlink capacity. 
In order to model the transmission rate's probability distribution, a learning algorithm, based on quantile regression (QR), is developed, and the proposed QR-DRL method is proved to converge to a stable distribution of downlink transmission rate. 
Simulation results show that, in the error-free CSI scenario, the proposed  approach yields over $30\%$ and $2$-fold increase in the downlink sum-rate, compared with a fixed IR reflection scheme and direct transmission scheme, respectively. 
Simulation results also show that by deploying more IR elements,  the downlink sum-rate can be significantly improved. 
However, as the number of IR components increases,  more time is required for channel estimation, and the  slope of increase in the IR-aided transmission rate will become smaller. 
Furthermore, under limited knowledge of CSI, simulation results show that the proposed QR-DRL method, which learns a  full distribution of the downlink rate,  yields a better prediction accuracy and improves the downlink rate by $10\%$ for online deployments, compared with a Q-learning baseline. 
  
\end{abstract}

{\small \emph{Index Terms} -- intelligent reflector;  reinforcement learning; multi-user MISO; millimeter wave; beyond 5G.}  
 
\IEEEpeerreviewmaketitle

\section{Introduction}

Next-generation cellular systems will inevitably rely on high-frequency millimeter wave (mmW) communications to meet the growing need for wireless capacity \cite{saad2019vision}.   
By leveraging the large bandwidth at mmW frequencies, a wireless network  can potentially deliver high-speed wireless links and meet stringent quality-of-service requirements.  
Moreover,  mmW  bands allow the implementation of small-sized antenna arrays and facilitate the use of massive multiple-input-multiple-output (MIMO) techniques to improve the wireless capacity.   
However, the high power cost and sophisticated signal processing of MIMO communications hinder the deployment of mmW frequencies into wide-scale commercial uses.  
Meanwhile, the high susceptibility to blockage caused by common objects, such as foliage and human bodies,  yields the uncertainty of mmW channels  \cite{wang2015multi}. 
Therefore, enabling reliable mmW links under blockage is a prominent challenge.

To overcome these intrinsic drawbacks of mmW, signal reflectors have been recently proposed to bypass obstacles and prolong the communication range \cite{peng2016effective}. 
By using reflectors, a non-line-of-sight (NLOS) mmW link can be compensated by creating multiple, connected line-of-sight (LOS) links \cite{zhang2019reflections}. 
Different from conventional relay stations (RSs) that receive, amplify (or decode), and forward the mmW signal, a reflector only reflects the incident signal towards the receiver, by inducing a phase shift. 
Therefore, the use of reflectors incurs no additional receiving noise or processing delay. 
Due to the nature of reflective surfaces, connected LOS links that pass through one or more reflectors can share the same frequency band, thus improving spectrum efficiency. 
A reconfigurable reflector can intelligently tune the conductivity of its metasurfaces, and, thus, reflect incident signals with different phase shifts. By jointly adjusting the phase shifts of a large number of such low-cost semi-passive elements, an intelligent reflector (IR) can focus the reflected signal into a sharp beam, hence maximizing the beamforming performance gain  with little energy cost.   
Indeed, it has been shown in \cite{tan2018enabling} that the use of reflectors is more appropriate for mmW networks than traditional RSs, in terms of energy, cost, and spectrum efficiency. 
However,  performing optimal reflection on the IR requires precise channel state information (CSI). 
Due to the possible mobility of the served user equipment (UE) and the blockage-prone nature of mmW signals, it is difficult for a practical IR-assisted wireless network to continuously obtain an accurate value for CSI.  
Thus, to enable a  real-time and efficient IR-aided transmission via mmW,  the challenges of CSI estimation and network performance optimization under imperfect CSI must be properly addressed.

\subsection{Related works}

The use of IRs to enhance the performance of cellular networks has attracted significant recent attention in  \cite{peng2016effective}, \cite{zhang2019reflections}, and \cite{ hu2018beyond, huang2019reconfigurable, wu2018intelligent, jamali2019reflect,di2019smart,huang2020holographic,zappone2019wireless, jung2019performance, taha2019enabling, nadeem2019intelligent, yang2019intelligent, mishra2019channel,wei2020parallel,alexandropoulos2020hardware}.   
In our previous work \cite{zhang2019reflections}, we studied the deployment of a UAV-carried IR whose goal is to optimize the downlink mmW transmission towards a mobile outdoor user, using a deep learning. 

In \cite{peng2016effective}, the design of a passive reflector and the estimation of the reflection gain are presented for indoor and urban mmW communications. 
The authors in \cite{hu2018beyond} investigated the potential of a large intelligent surface for positioning, and the work in \cite{huang2019reconfigurable} studied the energy efficiency of reflector-assisted downlink communications. 
The authors in \cite{wu2018intelligent} jointly optimized the transmit beamforming from an access point and the reflective beamforming in IR to maximize the received UE signal power. 
The work in \cite{jamali2019reflect} proposed a hybrid MIMO framework that applies reflective arrays and conventional transmit antennas to improve mmW energy efficiency.  
In \cite{di2019smart,huang2020holographic,zappone2019wireless}, the availability of  IR-aided transmissions and  learning-enabled communications is investigated for  practical network operations. 
However, the prior work on mmW reflectors in \cite{peng2016effective} focuses mainly on experimental measurements, while the IR-related works in \cite{zhang2019reflections} and  \cite{hu2018beyond,huang2019reconfigurable,wu2018intelligent,jamali2019reflect,di2019smart,huang2020holographic,zappone2019wireless} assumed perfect downlink CSI, which is challenging to know a prior in a practical network operation.

To obtain a precise value for CSI, recent works  \cite {jung2019performance, taha2019enabling, nadeem2019intelligent, yang2019intelligent,mishra2019channel,wei2020parallel,alexandropoulos2020hardware} have studied new approaches to efficiently measure CSI for IR-assisted communications.  
The authors in \cite{jung2019performance} and \cite{mishra2019channel} optimized the spectral efficiency and the downlink received power, respectively, for a large intelligent surface system with channel estimation error.  
In \cite{taha2019enabling} and \cite{nadeem2019intelligent}, a number of channel estimation protocols were investigated for reflecting beam training in a large intelligent surface communication system, based on deep learning and minimum mean squared error techniques, respectively.  
In order to reduce the overhead of channel estimation,  authors in \cite{yang2019intelligent}  aggregated adjacent reflective elements and measure the combined CSI for each group of IR components, and  authors in \cite{alexandropoulos2020hardware}  applied an active reception radio frequency chain to assist  channel estimation.
Furthermore, \cite{wei2020parallel} investigated the  problem of cascaded channel estimation by decomposing the channel information for each  wireless link. However, none of these prior works in \cite{jung2019performance, taha2019enabling, nadeem2019intelligent, yang2019intelligent, alexandropoulos2020hardware, wei2020parallel, mishra2019channel} studied the problem of IR-assisted cellular communications over mmW spectrum, which  is more sensitive to the real-time CSI, due to its shorter wavelength  and susceptibility to blockage.

\subsection{Contributions}

The main contribution of this paper is a novel framework for optimizing the IR-aided downlink transmission of a BS over mmW links.  
In particular, we propose a practical approach to estimate the downlink CSI, such that the reflection coefficient of the IR can be optimized in a real-time manner so as to maximize the downlink capacity of multiple UEs. Our main contributions are:

\begin{itemize}
	\item 
	First, under the assumption of perfect CSI, we jointly optimize the precoding transmission of the BS and the reflection coefficient of the IR.  
	An iterative algorithm is designed, based on Lagrangian transformation, fractional programming, and alternating optimization techniques,  to maximize  the downlink sum-rate of the IR-aided communications towards multiple users.	  
	\item  
	Given that receivers' noise induces errors to the measured CSI, we study the optimization of the IR-aided transmission with imperfect CSI.  
	In particular, we propose a distributional reinforcement learning (DRL) approach to model the distribution function of the downlink rate, and, then, the IR reflection coefficient is optimized to maximize the expected downlink capacity.  
	To model the rate's probability distribution, an iterative learning algorithm based on quantile regression (QR) is developed, so that the optimal reflection coefficient  is learned based on UEs' feedback. 
	We analytically prove that the proposed QR-DRL approach   converges to a stable distribution of IR-aided downlink transmission rates.

	\item 
	For a scenario with error-free CSI, simulation results show that, the proposed transmission approach outperforms two baselines: (i) a fixed IR reflection and (ii) a direct transmission scheme.  
	As the BS transmit power increases from $20$ to $40$ dBm, the proposed method yields over $30\%$  and $3$-fold increases in the average downlink rate, compared with both baselines.    
	Meanwhile, the proposed approach shows over $30\%$ and $2$-fold increases in performance, as the downlink bandwidth increases from $0.1$ to $3$ MHz.   
	As the number of transmit antennas increases from $16$ to $100$, the average sum-rates resulting all algorithms will increase, and the proposed method improves the data rate by over $30\%$, compared with the fixed IR scheme.  
	When the number of IR components increases, the performance of the direct transmission scheme remains the same, while both the IR-assisted methods yield higher rates. The proposed method improves the performance by over $20\%$, compared with the fixed IR scheme.  
    Moreover, simulation results show that as the number of IR components becomes larger,  more time is required for channel estimation, and, thus, the increase speed of the  IR-aided transmission rate will be slower.  
	
	\item
	For a scenario with imperfect CSI, simulation results show that the QR-DRL method, which learns a full distribution of the downlink sum-rate, has a slower convergence rate, but yields a better prediction accuracy and improves the average spectrum efficiency by over $10\%$ for online deployments, 
	compared with a  Q-learning baseline. 
	 
	
\end{itemize}

The rest of this paper is organized as follows. Section \ref{sec_systemModel} presents the system model. 
The transmission and reflection under perfect CSI is optimized in Section \ref{optPerfectCSI}. 
In Section \ref{optLimitedCSI}, the downlink sum-rate is maximized under limited knowledge of CSI, where a learning framework is proposed to tackle the uncertainty of downlink CSI and optimize the reflection downlink capacity.  
Simulation results are presented in Section \ref{sec_simulation} and conclusions are drawn in Section \ref{sec_conclusion}.  
 
\emph{Notation:} 
Italic letters $a$ and $A$ are both scalar, the bold letter $\boldsymbol{a}$ is a vector, the bold capital $\boldsymbol{A}$ is a matrix, and $a_{i,j}$ is the element on the $i$-th row and the $j$-th column of $\boldsymbol{A}$. 
The calligraphic capital $\mathcal{A}$ denotes a set. 
The blackboard bold $\mathbb{R}$, $\mathbb{C}$ and $\mathbb{N}^+$ are the sets of real numbers, complex numbers, and positive integer numbers, respectively. 
$\mathcal{CN}(\boldsymbol{\mu},\boldsymbol{\Sigma})$ denotes a complex normal random variable with mean vector $\boldsymbol{\mu}$ and covariance matrix $\boldsymbol{\Sigma}$.  
$\|\boldsymbol{A} \|_F$ is the Frobenius norm of the matrix $\boldsymbol{A}$, and  $\tr(\boldsymbol{A})$ denotes the trace. 
$\|\boldsymbol{a}\|$ is the two-norm of a vector. 
$|a|$, $\angle a$, and $\text{Re}\{a\}$ are the absolute value, angle and real part of the complex number in polar coordinates.    
$(\cdot)^T$, $(\cdot)^{H}$,  $(\cdot)^{-1}$, and  $(\cdot)^{+}$ are transport, Hermitian (conjugate transport), inverse, and pseudo-inverse, respectively.  
$\inf$ and $\sup$ are the infimum and supremum of a set. 
$\text{diag}(\boldsymbol{a})$ is a diagonal matrix with the entries of $\boldsymbol{a}$ on its diagonal.  
$\mathbbm{1}_{a=0}$ denotes an indicator function that equals to $1$  if ``$a=0$" is true,  $\boldsymbol{1}_n$ is an $n \times 1$ all-ones vector, and $\boldsymbol{I}_n$ is the $n \times n$ identify matrix.

\section{System Model}\label{sec_systemModel}
 
\begin{figure}[!t]
	\begin{center}
		\vspace{-0.8cm}
		\includegraphics[width=9.2cm]{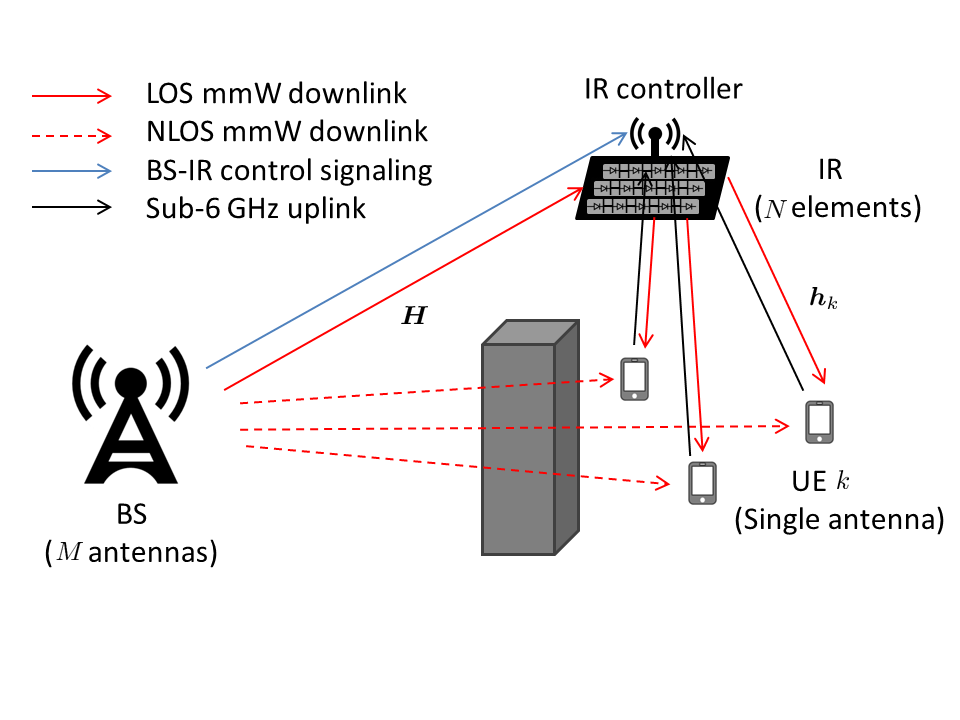}
		\vspace{-1.2  cm}
		\caption{\label{systemmodel}\small 
			The multi-user downlink communications from a multi-antennas BS, reflected via an IR, to a group of NLOS UEs, over mmW links.  
		}  
	\end{center}\vspace{-1  cm}
\end{figure}

Consider a cellular base station (BS) that provides mmW  downlink  communication to a set $\mathcal{K}$  of  $K$ wireless users. 
We assume that the served users are located within a hotspot area, and each UE is equipped with a single antenna\footnote{IR-aided multiple-input-multiple-output communication for multiple-antenna UEs will  be subject to our future work.}. 
In order to compensate for the fast attenuation of mmW signals, the BS is equipped with $M \in \mathbb{N}^{+}$ directional antenna arrays for  beamforming.  Therefore, in a LOS case, downlink communication can be reliable and efficient.  
However,  the BS-UE link can be blocked by common objects, such as buildings and foliage, which can seriously attenuate mmW signals.  
In particular, due to the high density of UEs within a hotspot area, the body of the human user becomes a main blockage source. 
Therefore, as done in \cite{huang2019reconfigurable} and \cite{nadeem2020asymptotic}, we assume that the direct link from the BS to each hotspot UE is always NLOS, and, thus, the received signal via the direct link is negligible.  

To overcome mmW blockage and improve the received UE power, an IR can be deployed to assist the mmW downlink transmission.   
As shown in Fig. \ref{systemmodel}, an IR can potentially replace one direct NLOS link with two connected LOS links, by reflecting the mmW signals from the BS towards each served UE.        
We assume that the IR consists of $N \in \mathbb{N}^{+}$  reflective components.  
By controlling the conductivity of its metasurfaces, the reflection coefficient of each IR element can be dynamically adjusted.    
Note that an IR is a passive device, which cannot sense any CSI or process received signals.  
To enable  information exchange, an active antenna must be embedded onto the IR controller to receive control signal from the BS and feedback information from  served UEs. 
As shown in Fig. \ref{systemmodel}, the BS assigns a specific channel to the IR  for control signaling.  
Meanwhile, the UE's feedback can be obtained via the uplink.  
Note that, this active antenna at the IR controller only receives and processes signals for control purposes, while downlink communication signals are transferred  through the reflection component at the IR. 
Moreover, we assume that the BS applies directional transmissions with a narrow beam, such that  the BS-IR channel has only one single path. Therefore, each UE will receive mmW signals all over the IR-related channels.

\subsection{Communication model} \label{commModel}

Our system can be viewed as a multi-user multiple-input-single-output (MISO) communication model,  
in which the BS transmits information to downlink UEs via a common frequency band, while being assisted by an IR \cite{huang2019reconfigurable}.  
The BS-IR channel  is denoted as $\boldsymbol{H} \in \mathbb{C}^{N\times M}$, and the IR-UE link for each UE $k \in \mathcal{K}$ is given as $\boldsymbol{h}_k \in \mathbb{C}^{1\times N}$.   
In order to provide downlink communications to multiple UEs, the BS precodes the transmit signal as an $M\times 1$ vector $\boldsymbol{s} = \sum_{k=1}^{K} \boldsymbol{w}_k \beta_k$, where  $\boldsymbol{w}_k \in \mathbb{C}^{M\times 1}$ is the precoding vector,  and $ \beta_k$ is the unit-power  information symbol  for UE $k$.   
Here, the power allocation is subject to a maximum power constraint $P_{\textrm{max}}$ of the BS, where $\| \boldsymbol{s} \|^2 = \sum_{k \in \mathcal{K}}  \|\boldsymbol{w}_k \|^2 \le P_{\textrm{max}}$. 

We consider that all IR components are equally spaced in a two-dimensional plane to form the IR. 
Let  $\mathcal{N}$ be the index set of $N$ IR elements. 
For each component $n \in \mathcal{N}$, we denote the phase shift of the reflection by $\theta_n \in [0,2\pi) $ and the amplitude reflection coefficient by $a_n  \in [0,1]$.  
The independent control of the phase  $\theta_n$ and amplitude $a_n$ for each IR component  $n \in \mathcal{N}$ was designed and experimentally characterized  in \cite{dai2018independent} and \cite{jia2016broadband}. 
Consequently, the IR's reflection coefficient will be given by ${\Theta} = \text{diag} (a_1  e^{j\theta_1},\cdots, a_N e^{j\theta_N} )$.  
Therefore, after the reflection, the received signal at UE $k$ will be: $	y_k =  \boldsymbol{h}_k \Theta \boldsymbol{H} \boldsymbol{s} + z_k$, 
where $z_k \sim \mathcal{CN}(0,\sigma^2)$ is the receiver noise at UE $k$. 
In order to separate the mmW propagation environment with the IR’s reflection, we rewrite the received signal at UE $k$ equivalently
in the following form: \vspace{-0.2cm} 
\begin{equation}
	y_k = \boldsymbol{\phi} \boldsymbol{G}_k  \boldsymbol{s} + z_k,  
\end{equation}   
where $\boldsymbol{\phi} = [a_1e^{j\theta_1},\cdots, a_N e^{j\theta_N}] \in \mathbb{C}^{1 \times N}$ is a vector of the reflection coefficient and $\boldsymbol{G}_k = \text{diag} ( \boldsymbol{h_k}) \boldsymbol{H} \in \mathbb{C}^{N \times M} $ is the CSI of the connected BS-IR-UE link towards UE $k$ without any phase shift.    
For tractability, we assume that  UEs' locations remain unchanged during one coherence time of the mmW transmission.   
Therefore, the downlink signal-to-interference-and-noise ratio (SINR)  from the BS, reflected by the IR, to each UE $k \in \mathcal{K}$ can be given by, 
\begin{equation}\label{sinr}  
	\eta_k ( \boldsymbol{W},  \boldsymbol{\phi}) =	
	\frac{  |\boldsymbol{\phi} \boldsymbol{G}_k \boldsymbol{w}_k |^2 }{  \sum_{i\ne k, i\in \mathcal{K}}  |\boldsymbol{\phi} \boldsymbol{G}_k \boldsymbol{w}_i  |^2  + \sigma^2},    
\end{equation}  
where
$\boldsymbol{W} = [\boldsymbol{w}_1,\cdots,\boldsymbol{w}_K] \in \mathbb{C}^{M \times K} $ is a precoding matrix at the BS.  
Consequently, the downlink sum-rate that the IR-assisted  communication system provides to all UEs is  
\begin{align}\label{equCapacity}
  r(\boldsymbol{W},   \boldsymbol{\phi}) =  \sum_{k=1}^K b\log_2 \left( 1+  \eta_k( \boldsymbol{W},  \boldsymbol{\phi}) \right),  
\end{align}  
where $b$ is the downlink bandwidth.

\subsection{Channel Measurement} \label{channelMeasure}

Considering the susceptibility of mmW signals to blockage,  a simple body movement of the human user can substantially change the CSI  of the BS-IR-UE link for each UE.  
Therefore, a real-time  estimation is necessary to measure the accurate CSI value, so that the reflection coefficient can be optimally determined.  
Here, we let ${\mathcal{G}} = \{ \boldsymbol{G}_k \}_{\forall k\in\mathcal{K}}$ be the CSI set of the BS-IR-UE links of each downlink UE. 

The CSI of a wireless link is assumed to be constant within each coherence time slot $\tau = \lambda_f/ v_{\text{e}} $, which is the ratio of the carrier wavelength $\lambda_f$ to the UE's speed $v_{\text{e}}$. 
In a hotspot area,  UEs  are often confined to geographically constrained spaces.
As a result,  the speed of each UE can be very small, and the coherence time of the IR-UE communication links can be sufficiently long to support real-time channel measurement.  
Meanwhile, the works in  \cite{va2017inverse,va2016impact,va2015basic} show that using directional transmission can effectively reduce the  channel variation. Thus, the  narrow beamwidth of beamforming transmissions ensures  a slow change of the CSI over the  BS-IR link. 
Moreover, given that we consider a single-path channel, the small-scale fading caused by multipath propagation is negligible. 
Therefore, the BS-IR-UE channel will experience a longer coherence time which enables the real-time channel estimation. 

Given that next-generation cellular networks will use dense small cell deployments, 	we assume that the distances of the BS-IR and IR-UE links are short enough, such that the mmW signal, reflected by a single IR element, can be properly captured at its receiver.  
Since the IR is a passive device that cannot transmit or decode any signal,   a time division duplex approach is used whereby the BS-IR-UE  channels are estimated, by exploiting the channel reciprocity, using the uplink pilot signals from UEs \cite{nadeem2019intelligent}. 
Therefore, within one channel coherence time slot, three sequential phases\cite{yang2019intelligent} are employed to measure the  CSI: Uplink training phase, processing phase, and downlink transmission phase, as shown in Fig. \ref{cohenrenceTime}. 
\begin{figure}[!t]
	\begin{center}
		\vspace{-1.2cm}
		\includegraphics[width=12cm]{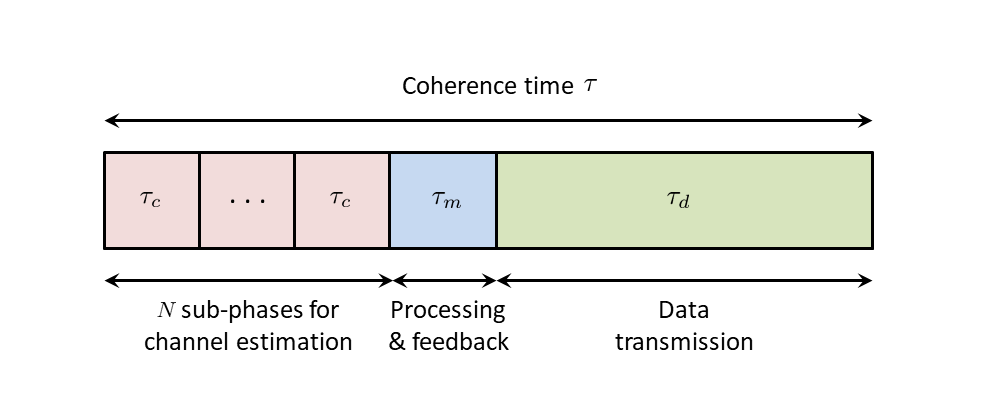} 
		\vspace{-0.4cm}
		\caption{\label{cohenrenceTime}\small One coherence time slot $\tau$ is divided into three sub-phases, where $\tau = N\tau_c + \tau_m + \tau_d$. } 
	\end{center}\vspace{-1 cm}
\end{figure}

Throughout the uplink training phase within a time duration  $N\tau_c$, each UE $k \in \mathcal{K}$ transmits mutually orthogonal pilot symbols ${s}_k$, where $|{s}_k|^2 = p_c$ is the transmit power,  
and an ON/OFF state control  is applied to the IR. An IR element $n \in \mathcal{N}$ in the ON state will reflect the incident signal without any phase shift, in which case we have $a_n=1$ and $\theta_n=0$; while an IR element in the OFF state will capture the incident signals without any reflection, i.e., $a_n=0$.    
Here, the uplink training phase is evenly divided into $N$ sub-phases,    
and within each sub-phase $\tau_c$, only one IR component  is ON\footnote{In the ON state, the  energy loss of signal reflection will be counted in the measurement of  CSI.}, while all the other IR components are OFF.   
Let $\boldsymbol{g}_{n,k} \in \mathbb{C}^{1\times M}$ be the $n$-th row vector of the channel matrix $\boldsymbol{G}_k$.
Then, the received training signal vector from the downlink UEs, reflected by the $n$-th IR element, to the BS can be expressed as  $\boldsymbol{y}_n = \sum_{k=1}^{K} \boldsymbol{g}^{H}_{n,k} s_k + \boldsymbol{z}$,  where $\boldsymbol{z} \sim \mathcal{CN}(\boldsymbol{0},\sigma^2_{\text{BS}} \boldsymbol{I}_M)$ is the noise vector at the BS.  
Once the uplink pilot phase is finished, the BS will receive $N$ independent observation vectors $\{ \boldsymbol{y}_n\}_{n=1, \cdots,N}$, and the signal processing phase with a fixed duration $\tau_m$  starts.    
Based on the  minimum-variance least-square estimation, the  channel vector from the BS, reflected by the $n$-th IR element, to UE $k$ can be calculated via    
\begin{equation}
	\hat{\boldsymbol{g}}_{n,k} = (\boldsymbol{y}_n s_k^{-1})^H \triangleq {\boldsymbol{g}}_{n,k} + \tilde{\boldsymbol{g}}_{n,k},
\end{equation} 
where $\tilde{\boldsymbol{g}}_{n,k} = ( \boldsymbol{z}  s_k^{-1})^H $ is the uncorrelated  estimation error.   
Therefore, the measurement results of the downlink channel matrix for each UE $k \in \mathcal{K}$ can be denoted by $\hat{\boldsymbol{G}}_k = [\hat{\boldsymbol{g}}^T_{1,k}, \cdots, \hat{\boldsymbol{g}}^T_{N,k}  ]^T \in \mathbb{C}^{N \times M}$. 
At the end of the  signal processing phases, the BS will send the  channel estimation  result $\hat{\mathcal{G}} =\{ \hat{\boldsymbol{G}}_k\}_{\forall k \in \mathcal{K}}$ to the IR, via the control channel.    
Once the processing phase is finished,  the transmission phase $\tau_d$ will start, during which the IR will provide the downlink service towards UEs, while optimizing the reflection coefficient to maximize the downlink rate in (\ref{equCapacity}).

However, the theoretical upper-bound of the transmission rate in (\ref{equCapacity}) is difficult to achieve in practice. 
This is because the optimization of the reflection coefficient $\boldsymbol{\phi}$ is based on the estimated CSI $\{\hat{\boldsymbol{G}}_k\}_{k=1,\cdots,K}$, which is subject to measurement errors.   
Here,  the mean square error (MSE) of the channel estimation for each BS-IR-UE MISO channel can be calculated by 
\begin{equation}
	\begin{aligned}
	\mathbb{E} \left\{\| \hat{\boldsymbol{G}}_k - \boldsymbol{G}_k \|^2_F \right\} 
	 = \mathbb{E} \left\{ \sum_{n=1}^N \| \tilde{\boldsymbol{g}}_{n,k} \|^2  \right\}  
	 = \sum_{n=1}^N  \mathbb{E} \left\{  \frac{\| \boldsymbol{z}  \|^2}{p_c}  \right\} = \frac{NM\sigma_{BS}^2}{p_c}. 
	\end{aligned}
\end{equation}
Given that the uplink power $p_c$ of each UE is strictly limited, as the numbers of BS antennas $M$ and  IR elements $N$ increase, the MSE of the estimated CSI will become larger.  
In order to guarantee an efficient reflection transmission, it is necessary to optimize the IR coefficient, using the feedback from downlink UEs.    
Hence, beyond the pilot-aided training at the beginning of the time slot, more information about  CSI can be acquired from UEs' feedback during the following transmission phase to improve downlink performance, as detailed in Section \ref{optLimitedCSI}.  

\begin{figure}[!t]
	\begin{center}
		\vspace{- 1  cm}
		\includegraphics[width=13.5cm]{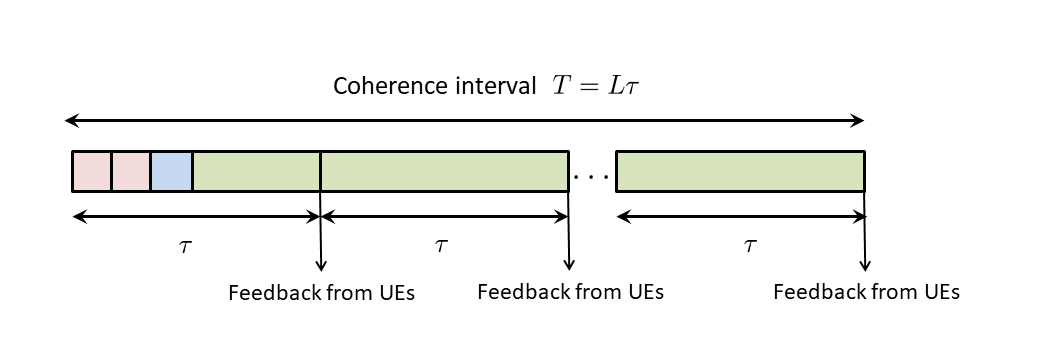} 
		\vspace{-0.4cm}
		\caption{\label{twotimescales}\small Multiple time slots are grouped into a coherence interval $T = L\tau$, during which the channel measurement is employed  only during the first $\tau$, while all the following steps are used for data transmission.}  
	\end{center}\vspace{-1 cm}
\end{figure}

Furthermore, we note that as the number of IR elements $N$ becomes larger, the effective duration $\tau_d = \tau - N\tau_c -\tau_m$ for downlink communications will decrease.  
In order to alleviate the heavy communication overhead of channel estimation, we group multiple time slots to form a longer coherence interval $T = L \tau$, where $L \in \mathbb{N}^{+}$, such that the channel measurement is employed only at the beginning of the first time slot $\tau$, while all remaining time steps are  used for downlink communication, as shown in Fig. \ref{twotimescales}.   
This  framework is supported by the experimental results in  \cite{va2016impact} and \cite{va2015basic}, where the coherence time of a wireless channel is shown to increase at least proportionally to the inverse of the beamwidth.
Therefore, given a narrow beam in the downlink transmission,   
we can consider the IR transmission over a sufficiently long interval $T$.  
In consequence, the total downlink transmission that the IR provides to the hotspot UEs within one coherence interval $T$ is 
\begin{equation}\label{aveCapacity}
	R( \boldsymbol{W}, \boldsymbol{\phi}  | \hat{\mathcal{G}}) =  \sum_{l=1}^{L} t_l \cdot r_l( \boldsymbol{W},  \boldsymbol{\phi}  | \hat{\mathcal{G}}), 	
\end{equation}
where $t_l$ is the transmission duration in  the $l$-th time slot, and  $r_l(\boldsymbol{W}, \boldsymbol{\phi}| \hat{\mathcal{G}}) $ is the downlink sum-rate defined in (\ref{equCapacity}). 
Here,  condition $\hat{\mathcal{G}}$  represents the imperfect knowledge that the BS has about the BS-IR-UE links when optimizing the beamforming and reflection coefficients. 
However, $\hat{\mathcal{G}}$ does not change any variables in the equation of $r_l$, and the channel matrix in $r_l$ is still the real physical  channel coefficient $\{ \boldsymbol{{G}}_k\}_{\forall k}$.   
Based on Fig. \ref{twotimescales}, we have $t_l = \tau_d$ for $l=1$ and $t_l = \tau$ for $l=2,\cdots,L$. 
Meanwhile, if the values of $\boldsymbol{W}$  and $\boldsymbol{ \phi }$ are identical for all $l = 1,\cdots, L$, then $r_l = r$ and  $R = (T -N\tau_c - \tau_m ) \cdot r \triangleq   T_d \cdot r$, where $T_d=T -N\tau_c - \tau_m$ is the effective transmission time during one coherence interval.  
Compared with the communication duration $\tau_d = \tau - N\tau_c -\tau_m$  from the channel estimation framework in Fig. \ref{cohenrenceTime}, the proposed framework in Fig. \ref{twotimescales} significantly prolongs the averaged  transmission time within each  coherence interval $T$.  

\subsection{Problem Formulation}
Our goal is to jointly optimize the precoding matrix at the BS and the reflection coefficient of the IR, such that the total  data transmissions that the IR provides to downlink UEs within a coherence interval can be maximized, i.e.: 
\begin{subequations}   \label{equOptall}  
	\begin{align} 
	\max_{  \boldsymbol{W},   \boldsymbol{\phi} }\quad &   R ( \boldsymbol{W},  \boldsymbol{\phi}  | \hat{\mathcal{G}})    \label{equOptbsall}\\
	\textrm{s. t.} \quad  
	& \sum_{k \in \mathcal{K}} \| \boldsymbol{w}  \|_k^2\le P_{\textrm{max}}, \label{conPowerPrecoder} \\  
	& |\phi_n| \le 1, \forall n \in \mathcal{N}.   \label{conReflection}  
	\end{align}
\end{subequations} 
The objective function in (\ref{equOptbsall}) is the summation of downlink data that the IR-aided transmission provides to all UEs within one $T$.  
(\ref{conPowerPrecoder}) is the  power limitation at the BS, and (\ref{conReflection}) is the reflection constraints at the IR.  
The optimization problem (\ref{equOptall}) is challenging to solve for two reasons.    
First, during each coherence time, the objective function (\ref{equOptbsall}) is non-convex  with respect to $\boldsymbol{W}$ and $\boldsymbol{\phi}$.     
Second,  uplink pilot training is subject to measurement errors, which leads to  inaccurate CSI and renders the IR-aided transmission less reliable.     
In order to address the aforementioned challenges, first, in Section \ref{optPerfectCSI}, we analyze a simple case in which the estimated CSI is error-free. 
In this section,  the beamforming transmission at the BS and the reflection coefficient of the IR will be jointly optimized, under perfect CSI. 
Next,  in Section \ref{optLimitedCSI},
a learning-based approach can be applied to capture the uncertainty of the downlink channel and enable the optimal reflection at the IR, under imperfect CSI.

\section{Optimal transmission and reflection with perfect CSI} \label{optPerfectCSI}  

In this section, the precoding matrix  at the BS and the reflection coefficient at the IR will be jointly optimized in order to maximize the downlink transmission capacity, under the assumption of perfect CSI.   
During each  time slot,  (\ref{equOptall}) is reduced to the following form: 
\begin{subequations}   \label{equOpt1}
	\begin{align} 
	\max_{  \boldsymbol{W}, \boldsymbol{\phi} } \quad &   r ( \boldsymbol{W}, \boldsymbol{\phi} ) \label{1equOpt}\\
	\textrm{s. t.} \quad  
	& \sum_{k \in \mathcal{K}} \| \boldsymbol{w}  \|_k^2\le P_{\textrm{max}}, \label{1conPowerPrecoder} \\    
	& |\phi_n| \le 1, \forall n \in \mathcal{N}.   \label{1conReflection}     
	\end{align}
\end{subequations} 
In order to solve this non-convex problem, we apply the Lagrangian dual transform method \cite{shen2018fractional}, by introducing an auxiliary variable $\boldsymbol{\alpha}$. 
Then, the objective function (\ref{1equOpt}) can be equivalently rewritten  as: 
\begin{equation}\label{sumrateAlpha}
r_{\alpha} (\boldsymbol{W} ,\boldsymbol{\phi}, \boldsymbol{\alpha}) = b\sum_{k=1}^K \left(\log_2 (1 + \alpha_k)   -\alpha_k + \frac{(1+\alpha_k) \eta_k}{1+\eta_k }  \right).
\end{equation} 
Thus, we equivalently reformulate the optimization problem (\ref{equOpt1}) into the following form:
\begin{subequations}   \label{equOptalpha}
	\begin{align} 
	\max_{  \boldsymbol{W}, \boldsymbol{\phi},\boldsymbol{\alpha} } \quad &   r_{\alpha} (\boldsymbol{W} ,\boldsymbol{\phi}, \boldsymbol{\alpha})  \label{alphaequOpt}\\
	\textrm{s. t.} \quad  
	& \sum_{k \in \mathcal{K}} \| \boldsymbol{w}  \|_k^2\le P_{\textrm{max}}, \label{alphaconPowerPrecoder} \\    
	& |\phi_n| \le 1, \forall n \in \mathcal{N}.    \label{alphaconReflection}     
	\end{align}
\end{subequations} 
Next, we  solve (\ref{equOpt1}) by solving its equivalent problem (\ref{equOptalpha}) using the alternating optimization method \cite{jain2013low}.  

First, by holding $\boldsymbol{W}$ and $\boldsymbol{\phi}$  fixed and setting $\frac{\partial  r_{\alpha} }{\partial  \alpha_k} = 0$, we can find the optimal value of $\alpha_k$ as $\alpha^o_k = \eta_k$. Then, for  a fixed $\boldsymbol{\alpha}$, the optimization problem  is reduced to
\begin{subequations}\label{equsOptsub2}
	\begin{align}
	\max_{  \boldsymbol{W} ,\boldsymbol{\phi} } \quad &   \sum_{k=1}^K  \frac{\hat{\alpha}_k \eta_k}{1+\eta_k}   \label{equOptsub2}\\
	\textrm{s. t.} \quad  & (\text{\ref{1conPowerPrecoder}}), (\text{\ref{1conReflection}}), 
	\end{align}
\end{subequations}
where $\hat{\alpha}_k = b (1+\alpha_k)$. 
Given that (\ref{equsOptsub2}) is a multiple-ratio fractional programming problem, we can fix the value of $\boldsymbol{W}$ and $\boldsymbol{\phi}$ alternatively, and solve the optimization problem via an iterative approach, detailed as follows.

\subsection{Optimal precoding matrix at BS with perfect CSI} 
For a fixed $\boldsymbol{\phi}$, the optimal precoding problem becomes
\begin{subequations}\label{equsOptsub3}
	\begin{align}
	\max_{  \boldsymbol{W}   } \quad &  f(\boldsymbol{W}) = \sum_{k=1}^K \hat{\alpha}_k  \frac{  |\boldsymbol{\phi} \boldsymbol{G}_k  \boldsymbol{w}_k |^2 }{  \sum_{i=1 }^{K} |\boldsymbol{\phi} \boldsymbol{G}_k  \boldsymbol{w}_i  |^2  + \sigma^2}  \label{equOptsub3}\\
	\textrm{s. t.} \quad  & \sum_{k=1}^{K}\|\boldsymbol{w}_k \|^2 \le P_\textrm{max}. 
	\end{align}
\end{subequations} 	
The multiple-ratio fractional programming function in (\ref{equOptsub3}) is equivalent to 
\begin{equation}
	  f_{\lambda}(\boldsymbol{W},\boldsymbol{\lambda}) =   - \sum_{k=1}^K |\lambda_k|^2  (\sum_{i=1}^K | \boldsymbol{\phi} \boldsymbol{G}_k  \boldsymbol{w}_i |^2 + \sigma^2)+\sum_{k=1}^K   2 \sqrt{\hat{\alpha}_k} \text{Re} \{ \lambda_k  \boldsymbol{\phi} \boldsymbol{G}_k  \boldsymbol{w}_k  \} ,
\end{equation}   
where  $\boldsymbol{\lambda} \in \mathbb{R}^{K \times 1}$ is an auxiliary vector \cite[Corollary 2]{shen2018fractional}. 
Since $f_{\lambda}$ is a convex function with respect to both $\boldsymbol{w}_k$ and ${\lambda}_k$, $\forall k$, an iterative approach can be applied to optimize  $\boldsymbol{W}$ and $\boldsymbol{\lambda}$ alternatively. 
First, we fix the value of $\boldsymbol{W}$ and set $\frac{\partial f_{\lambda} }{\partial \lambda_k}=0$. Then, the optimal value of $\lambda_k$ is 
\begin{equation}\label{zerolambda}
\lambda^o_k = \frac{\sqrt{\hat{\alpha}_k}  \boldsymbol{\phi} \boldsymbol{G}_k  \boldsymbol{w}_k}{\sum_{i=1}^K | \boldsymbol{\phi} \boldsymbol{G}_k  \boldsymbol{w}_i |^2 + \sigma^2 }.
\end{equation} 
Then, by fixing $\boldsymbol{\lambda}$, the optimal $\boldsymbol{w}_k$ can be given by 
\begin{equation} \label{zerow} \vspace{-0.1cm}
\boldsymbol{w}_k^o = \sqrt{\hat{\alpha}_k} \lambda_k \boldsymbol{\phi} \boldsymbol{G}_k \left( \kappa_o \boldsymbol{I}_M + \sum_{i=1}^K |\lambda_i|^2 (\boldsymbol{\phi} \boldsymbol{G}_i )(\boldsymbol{\phi} \boldsymbol{G}_i )^H    \right)^{-1}   ,
\end{equation}
where $\kappa_o \ge 0$ is the minimum value such that $\sum_{k=1}^K\| \boldsymbol{w}_k^o \|^2 \le P_{\text{max}}$ holds. 
By alternating between (\ref{zerolambda}) and (\ref{zerow}),  the value of the objective function in (\ref{equOptsub3}) will increase. 

\subsection{Optimal reflection at IR with perfect CSI}

Next, we fix the value of $\boldsymbol{W}$ and optimize the reflection coefficient $\boldsymbol{\phi}$ in (\ref{equsOptsub2}), i.e.:    
\begin{subequations}\label{equsOptsub4} \vspace{-0.1cm}
	\begin{align}
	\max_{  \boldsymbol{\phi}   } \quad &  f(\boldsymbol{\phi}) = \sum_{k=1}^K \hat{\alpha}_k  \frac{  |\boldsymbol{\phi} \boldsymbol{G}_k  \boldsymbol{w}_k |^2 }{  \sum_{i=1 }^{K} |\boldsymbol{\phi} \boldsymbol{G}_k  \boldsymbol{w}_i  |^2  + \sigma^2}  \label{equOptsub4}\\
	\textrm{s. t.} \quad  &  |\phi_n | \le 1,  \quad \forall n \in \mathcal{N}.  
	\end{align}
\end{subequations}
Similarly, an auxiliary vector $\boldsymbol{\delta}$ is introduced, so that (\ref{equOptsub4}) equivalently becomes  
$f_{\delta} (\boldsymbol{\phi},\boldsymbol{\delta}) =  \sum_{k=1}^K   2 \sqrt{\hat{\alpha}_k} \text{Re} \{ \delta_k  \boldsymbol{\phi} \boldsymbol{G}_k  \boldsymbol{w}_k  \} $ $ - \sum_{k=1}^K |\delta_k|^2  (\sum_{i=1}^K | \boldsymbol{\phi} \boldsymbol{G}_k  \boldsymbol{w}_i |^2 + \sigma^2)$, which is convex with respect to both $\boldsymbol{\phi}$ and $\boldsymbol{\delta}$.   
Hence, the value of the reflection coefficients $\boldsymbol{\phi}^*$ can be optimized via an approach that is similar to the optimization process used for the precoding  matrix.   
First, the optimal $\delta_k$ for a given $\boldsymbol{\phi}$ is 
\begin{equation}\label{zerodelta}
\delta_k^o = \frac{\sqrt{\hat{\alpha}_k}  \boldsymbol{\phi} \boldsymbol{G}_k  \boldsymbol{w}_k}{\sum_{i=1}^K | \boldsymbol{\phi} \boldsymbol{G}_k  \boldsymbol{w}_i |^2 + \sigma^2 }. 
\end{equation}
Then, the optimization of the reflection coefficient for a fixed $\delta_k$ becomes
\begin{subequations}\label{equsOptsub5}\vspace{-0.2cm}
	\begin{align}
	\max_{  \boldsymbol{\phi}   } \quad &  f(\boldsymbol{\phi}) =  -\boldsymbol{\phi} \boldsymbol{U} \boldsymbol{\phi}^H + 2\text{Re}\{ \boldsymbol{\phi} \boldsymbol{v} \} -C  \label{equOptsub5}\\
	\textrm{s. t.} \quad  &  |\phi_n | \le 1,  \quad \forall n \in \mathcal{N},   \label{cons}
	\end{align}
\end{subequations}
where  $C = \sum_{k=1}^K |\delta_k|^2\sigma^2$,  $\boldsymbol{v} = \sum_{k=1}^K {\delta_k}^H \boldsymbol{G}_k \boldsymbol{w}_k$, and $\boldsymbol{U} = \sum_{k=1}^K |\delta_k|^2 \sum_{i\ne k}  \boldsymbol{G}_k  \boldsymbol{w}_i ( \boldsymbol{G}_k  \boldsymbol{w}_i)^H $. 
Given that $\boldsymbol{U}$ is a positive-definite matrix,   $ f(\boldsymbol{\phi})$ is quadratic concave with respect to $\boldsymbol{\phi}$. 
Meanwhile,  constraint (\ref{cons}) is a convex set.
Thus, (\ref{equsOptsub5}) is solvable using a Lagrange dual decomposition \cite{guo2019weighted}. 

Therefore, given the CSI $\boldsymbol{G}_k$ of each BS-IR-UE link,  the precoding matrix $\boldsymbol{W}^*$ and reflection coefficient $\boldsymbol{\phi}^{*}$ can be optimized, based on Algorithm \ref{algo1}, to improve the downlink transmission sum-rate. 
Once (\ref{equOpt1}) is  solved, the BS will send the optimal solution $\boldsymbol{W}^*$ and $\boldsymbol{\phi}^{*}$  to the IR via the control channel. 
Here, we stress that Algorithm \ref{algo1} is derived under the assumption of perfect CSI. 
However, in presence of  channel estimation errors, Algorithm \ref{algo1} cannot guarantee an efficient downlink transmission.     
Therefore, under limited CSI, it is necessary for the IR to optimize the reflection parameter coefficient $\boldsymbol{\phi}$, based on the feedback of downlink UEs.

\begin{algorithm}[!t] \small  
	\caption{Optimal precoding  and reflection coefficients with perfect CSI} \label{algo1}
	\begin{algorithmic}
		\State \textbf{Initialize} the reflection coefficient $\boldsymbol{\phi} $  and the precoding matrix $\boldsymbol{W}$ \\ 
		\textbf{Repeat}  \\ 
		\quad 1. Update the auxiliary variable $\alpha_k = \eta_k$, $\forall k \in \mathcal{K}$; \\
		\quad 2. Optimize $\boldsymbol{W} $ by alternating between (\ref{zerolambda}) and (\ref{zerow}); \\ 
		\quad 3. Optimize $\boldsymbol{\phi} $  by alternating between (\ref{zerodelta}) and (\ref{equsOptsub5}); \\
		\textbf{Until} the value of $r_{\alpha}$ in (\ref{sumrateAlpha}) converges.  
	\end{algorithmic}
\end{algorithm}

\section{Optimal reflection with  limited CSI}
\label{optLimitedCSI}

In this section, the reflection optimization of the IR is studied, considering the estimation error in the measured CSI $\hat{\mathcal{G}}$, while Algorithm \ref{algo1} is used as initialization for the IR-aided transmission.  

\subsection{Problem formulation} 

During the communication stage, at the end of each coherence time $l$, the IR can get feedback from each UE $k$ about the downlink transmission performance.  
If the received signal power $|y_{l,k}|^2$ for each UE $k$ is small, then the optimal reflection from Algorithm \ref{algo1} is actually not reliable. 
In this case, the IR controller can adjust the reflection coefficient,  to improve the downlink transmission rate in a long term, based on the UEs' feedback.  
The optimization problem for the reflection adjustment of the IR at the end of each time slot $l = 1,\cdots,L-1$ is given as:
\begin{subequations} \label{equOpt_b} 
	\begin{align}
	\max_{ \Delta \Theta} \quad &   \sum_{t=l+1}^{L}  r_{t}(\boldsymbol{\phi}_{l} \times \Delta \Theta  |  \boldsymbol{y}_{l}, \hat{\mathcal{G}},  \boldsymbol{W}^* )   \label{equOptuavTau_b}  \\   
	\textrm{s. t.} \quad  
	& | \Delta \theta_n| = 1,  \forall n \in \mathcal{N},    \label{equConstrainIR_b}
	\end{align}
\end{subequations}  
where  $\boldsymbol{\phi}_l = [{\phi}_{l,1},\cdots,{\phi}_{l,N}]$ is the reflection coefficient during the $l$-th time slot, $ \Delta {\Theta} = \text{diag}( \Delta \theta_1,$ $ \cdots, \Delta \theta_N ) \in \mathbb{C}^{N \times N}$ is an adjustment matrix that shifts the phase of the reflection parameter $\boldsymbol{\phi}_l$ via $\boldsymbol{\phi}_l \times \Delta {\Theta}$ while keeping the reflection amplitude the same, and 
$\boldsymbol{y}_l  = [y_{l,1},\cdots, y_{l,K}]   \in \mathbb{C}^{1 \times K}$ is the received signal vector of downlink UEs.  
Objective function (\ref{equOptuavTau_b}) is the sum of downlink data rates of all UEs from the $l+1$-th time slot to the end of the coherence interval, given limited  CSI.  
(\ref{equConstrainIR_b}) is the  amplitude constraints of the reflection adjustment.   
The overall  process of the IR-aided transmissions under imperfect CSI is summarized in  Algorithm \ref{algo2}. 
The  value of $ \boldsymbol{W}^*$ remains the same during the transmission phase, and we will omit it in the following discussion.  

\begin{algorithm}[!t]  \small   
	\caption{Overall process of the IR-aided transmission under imperfect CSI} \label{algo2}
	\begin{algorithmic}
		\State  
		1. \textbf{Uplink training phase}:  estimate the downlink CSI $\hat{\mathcal{G}}$  as elaborated in Section \ref{channelMeasure}; \\ 
		2. \textbf{Processing phase}: optimize the precoding matrix $\boldsymbol{W}^*$ and reflection parameter $\boldsymbol{\phi}^*$ based on Algorithm \ref{algo1}; \\  
		3. \textbf{Transmission phase}: At the end of each time slot $l = 1,\cdots, L-1$, \\
		\quad A. Each UE $k$ updates its received signal $y_{l,k}$ via uplink;\\
		\quad B. The IR controller adjusts the reflection parameter $\boldsymbol{\phi}$ in (\ref{equOpt_b}), based on Algorithm \ref{algo3}  in Section \ref{secDRL}.    
	\end{algorithmic}
\end{algorithm}

The optimization problem  (\ref{equOpt_b}) is very challenging to solve for two reasons.  
First, due to the receiver's noise, it is impossible to have a perfect CSI via uplink training phase. 
One traditional approach to cancel the effect of noise is to average the estimation results from multiple measurements.    
However, repeated  measurements will consume significant time and power, which is impractical in the considered problem. 
Second, for each UE $k$, the scalar feedback $y_{l,k}$ during each time slot cannot provide sufficient information for the IR to reconstruct the channel matrix $\boldsymbol{G}_k$ which has $N\times M $ scalar elements.   
Therefore, a learning-based framework is introduced to improve the reflection parameter of the IR during the service of downlink transmission, in which the BS and IR are not required to have perfect measurement on the CSI or have an explicit knowledge of  channel statistics. 
By observing the UEs’ feedback during each time slot, the proposed algorithm can automatically acquire the real-time downlink statistics. 

\subsection{Reinforcement learning framework}  

In order to improve the reflection coefficients of the IR for each time slot, a distributional reinforcement learning (DRL) framework is designed to capture the measurement errors in the estimated CSI and model the relationship between the reflection coefficients and the downlink communication performance. 
As shown in Fig. \ref{RLmodel}, in our DRL framework,  the connected channels $\{\boldsymbol{ \phi } \boldsymbol{G}_k \}_{ \forall k}$ represent the communication \emph{environment},  
the IR controller is  the \emph{agent} that takes action $\Delta \Theta$ to change the reflection coefficient from $\boldsymbol{\phi}$ to $\boldsymbol{\phi} \times \Delta \Theta$, and 
the communication \emph{state} is a deviation vector $\boldsymbol{e} = [e_{1},\cdots,e_{K}] \in \mathbb{C}^{1 \times K}$, where for each UE $k$, the signal deviation $	e_{k} = y_{k} -\hat{y}_{k} = \boldsymbol{\phi}(\boldsymbol{G}_k - \hat{\boldsymbol{G}}_k) \boldsymbol{s} + z$ evaluates the accuracy of the measured downlink CSI $\{\hat{\boldsymbol{{G}}}_k\}_{\forall k}$. 
At the end of each time slot, the IR receives a communication reward $r$, which is defined as the downlink sum-rate towards all UEs  in (\ref{equCapacity}). 
Here, due to the small-scale fading of mmW channels, the downlink CSI may vary between different time slots, even for a fixed reflection coefficient of the IR. Thus, it is more suitable to consider the reward $r$ as a random variable with respect to each communication state $\boldsymbol{e}$ and reflection action $\Delta \Theta$, rather than a determined value.   
Meanwhile, $P(\boldsymbol{e}^{'}|\boldsymbol{e},\Delta \Theta)$ is the transition probability of the state from $\boldsymbol{e}$ to $\boldsymbol{e}^{'}$ after taking action $\Delta \Theta$.

\begin{figure}[!t]
	\begin{center}
		\vspace{- 1.5 cm}
		\includegraphics[width=12cm]{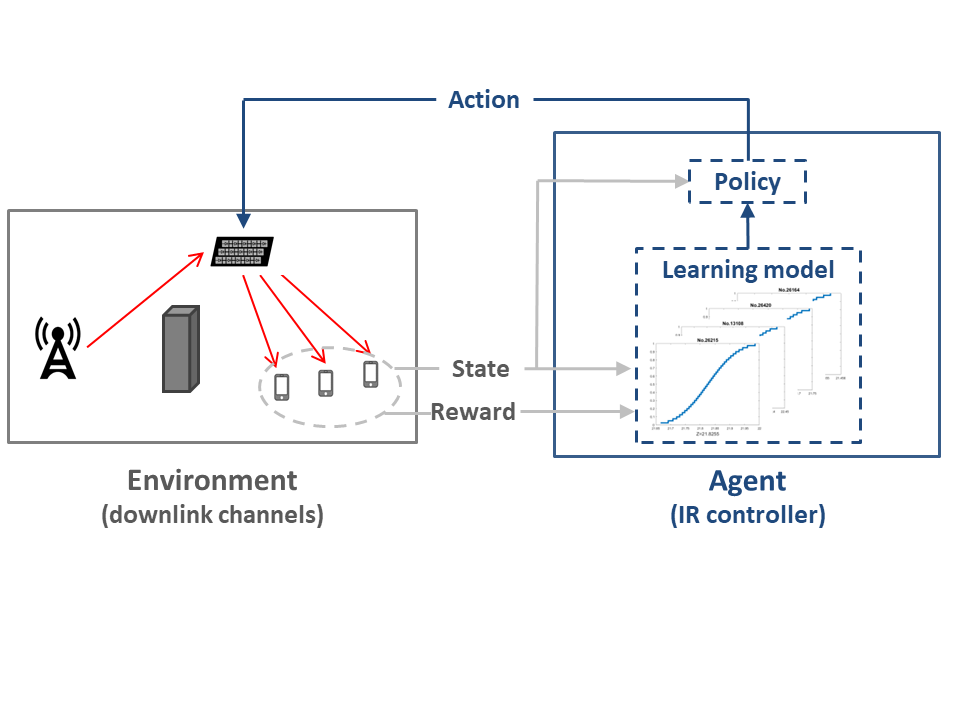} \vspace{- 2.2 cm}
		\caption{\label{RLmodel}\small An illustration of the RL framework for learning and optimizing the IR-aided downlink transmissions. }  
	\end{center}\vspace{-1 cm}
\end{figure}

At the end of each time slot $t$, after observing a deviation vector $\boldsymbol{e}_l$, the IR aims to adjust its reflection coefficient to accommodate the actual downlink CSI via $\boldsymbol{\phi}_{l+1}=\boldsymbol{\phi}_l \times \Delta \Theta$, such that the sum of future data rate $r_{t}$ for the following time slot $t = l+1,\cdots,L$ can be maximized. 
In order to quantify  the potential of each action matrix $\Delta \Theta$ to improve the future reward, first, a \emph{policy} $\pi(\Delta \Theta|\boldsymbol{e})$ is introduced to define the probability that the IR controller will adjust the reflection coefficient by $\Delta \Theta$, under a current state  $\boldsymbol{e}$. 
Then, for each stationary policy $\pi$, the potential of each state-action pair $(\boldsymbol{e}_l,\Delta \Theta_l)$ on improving the downlink transmission rate is defined by the sum of the discounted rewards, as follows:  
\begin{equation}\label{return}  
\begin{aligned}
Z^{\pi}(\boldsymbol{e}_l,\Delta \Theta_l)  = \sum_{t=l+1}^{\infty} \gamma^{t-l}~ r(\boldsymbol{e}_t,\Delta \Theta_{t}|   \boldsymbol{y}_l,\hat{\mathcal{G}}), 
\end{aligned}	 
\end{equation}
where $\Delta \Theta_{t} \sim \pi(\cdot|\boldsymbol{e}_{t})$,  $\boldsymbol{e}_{t+1} \sim P(\cdot|\boldsymbol{e}_{t}, \Delta \Theta_t )$, and $\boldsymbol{\phi}_{t+1} = \boldsymbol{\phi}_{t} \times \Delta \Theta_t$.   
Here, $\gamma \in (0,1)$  discounts the future rewards in the current estimation for each state-action pair. 
If $ \gamma \rightarrow 1$ and $L \rightarrow  \infty$,  the return function in (\ref{return}) approximates the objective function in (\ref{equOptuavTau_b}).  
Meanwhile,  if the CSI is assumed to be constant in each time slot, then, by fixing the value of the reflection coefficient $\boldsymbol{\phi}$ to be $\boldsymbol{\phi}_t^{*}$,  (19a) can be rewritten as $R_t= \sum_{t=l+1}^{L} r(\boldsymbol{\phi}_t^{*})$, and (20) becomes $Z_t = \sum_{t=l+1}^{\infty} \gamma^{t-l} r(\boldsymbol{\phi}_t^{*})$. 
Given that $\gamma$ is constant,  maximizing $Z_t$ with respect to $\boldsymbol{\phi}$ will be equivalent to the optimization of $R_t$. 
Thus, (20) provides an accurate framework for the IR controller to determine the efficient reflection coefficient   $\boldsymbol{\phi}_t^{*}$ at the time slot $t$,  such that the sum of the downlink data rates in the following time slots can be maximized.   
Thus, this cumulative discounted reward in (\ref{return}) is called the \emph{return value} that the IR can achieve by adjusting the reflection coefficient to $\boldsymbol{ \phi }_l \times \Delta \Theta$ fro the next communication slot.  
Meanwhile, given that $r$ is a random variable, it is necessary to model a distribution function of (\ref{return}) to identify the return value for each state-action pair.  
Once the return distribution is known, the optimal policy  $\pi$ that maximizes the expectation of the cumulative rewards can be defined by
\begin{equation} \label{max_return}
	\Delta \Theta^{*}_{l} = \arg\max_{\Delta \Theta} Z^{\pi}(\boldsymbol{e}_l,\Delta \Theta).
\end{equation} 
Thus, the optimal reflection parameter for the next time slot will be  $\boldsymbol{\phi}_{l+1} = \boldsymbol{\phi}_{l} \times \Delta \Theta^{*}_l$. 


Here, we note that a DRL framework is different from conventional reinforcement learning approaches, such as deep Q-learning \cite{mnih2013playing}, because a traditional RL method predicts the future return by a scalar, while, in our work, we aim to account for the uncertainty of the reward function, by modeling the sum-rate ${r}$ as a distribution function.  
Given that the CSI  error $\tilde{\boldsymbol{g}}_{n,k}$ in (4) is a random variable, 
the  precoding and reflection coefficients that  are optimized using the imperfect CSI will result in a random downlink sum-rate. 
Thus, given each choice of a reflection coefficient, the corresponding reward $r$ can be considered to be a random variable.   
Compared with a scalar value, a return distribution provides more information on the IR-assisted transmission features, and, thus, it enables more accurate predictions on the future reward.     
Here, we stress that the distribution is not designed to capture the uncertainty in the estimation of the reward, but rather to model the intrinsic randomness in the interaction between the IR  and  communication environment, due to the small-scale of mmW channels and the existence of noise in the channel estimation and information exchange phases. 
The IR-aided communication has intrinsic  uncertainty in its downlink rate, and, thus,  
it is more suitable to consider the downlink sum-rate ${r}$ as a random variable, with respect to each state-action pair $(\boldsymbol{e},\Delta \Theta)$, rather than a fixed number.   

\subsection{Distributional reinforcement learning with quantile regression} \label{secDRL}

 
In order to model the return distribution $Z^{\pi}$ for each state-action pair, a quantile regression (QR) method \cite{dabney2018distributional} is applied. 
A Q-quantile model $Z_Q$  approximates the target distribution $Z^{\pi}$ by a discrete  function with variable locations of $Q$ supports and fixed quantile of $\frac{1}{Q}$ probabilities, for a fixed integer $Q \in \mathbb{N}^+$.   
Mathematically, a Q-quantile model is denoted by  $Z_Q =[ z_1(\boldsymbol{e},\boldsymbol{\phi}), \cdots,$ $ z_Q(\boldsymbol{e},\Delta \Theta)  ] $, with 
a cumulative probability $F_{Z_Q}(z_q ) = \frac{q}{Q}$ for $q = 1,\cdots,Q$.  
Therefore, for each state-action $(\boldsymbol{e},\Delta \Theta) $, the objective is to find the optimal locations of each support in $Z_Q$ minimizes the ``distance'' between the target $Z^{\pi}$ and the Q-quantile model $Z_Q$ can be minimized. 
Here, 1-Wasserstein metric is used to define the ``distance'' between two distributions by
\begin{equation} \label{Wmetric}
	d_1 (Z_Q,Z) =  \int_{0}^{1} |F^{-1}_{Z_Q}(\omega) - F^{-1}_{Z}(\omega)|\dif \omega  ,  
\end{equation}
where  $F^{-1}_X$ is the inverse cumulative probability of distribution $X$.  
Based on Lemma 2 in \cite{dabney2018distributional}, the optimal value of each support $z^{*}_q$ that minimizes the 1-Wasserstein distance in (\ref{Wmetric}) is $z^{*}_q  = F_Z^{-1}(\frac{2q-1}{2Q})$. 
Fig. \ref{FourQuantiles} shows an example of 1-Wasserstein minimizing projection onto a $4$-quantile estimation of a target distribution. 

\begin{figure}[!t]
	\begin{center}
		\vspace{- 1  cm}
		\includegraphics[width=9cm]{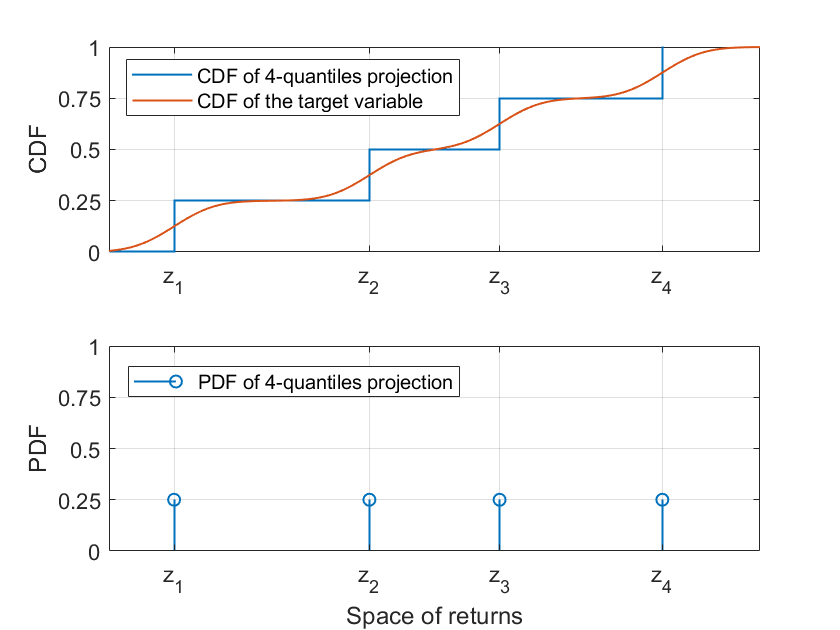} \vspace{- 0  cm}
		\caption{\label{FourQuantiles}\small An example of 1-Wasserstein minimizing projection  of a target distribution $Z$ (red line), with 4-quantiles $[z_1,z_2,z_3,z_4]$ (blue line).}  
	\end{center}\vspace{-1   cm}
\end{figure}

However, in the considered problem, the return distribution $Z^{\pi}$ of the downlink sum-rate is not explicitly known.  
Thus, the optimal approximation result  $z^{*}_q  = F_Z^{-1}(\frac{2q -1 }{2Q})$ is not directly available. 
In order to approximate the return distribution,  an empirical distribution $\hat{Z} \sim Z^{\pi}$ will be formed, based on the transmission feedback of $\boldsymbol{e}_l$, $\Theta_l$  and $r_{l+1}$ during each time slot, and $ {Z}$ is used as the target distribution to model the return  approximation $Z_Q$. 
In order to evaluate the approximate accuracy,  
we define the quantile regression loss  between $ {Z}$ and the Q-quantiles approximation $Z_Q$, as \cite{dabney2018distributional}
\begin{equation}\label{QRloss}
	\mathcal{L}_Z(Z_Q) = \sum_{q=1}^{Q} \mathbb{E}_{Z \sim Z^{\pi}} \left[ |\omega_q - \mathbbm{1}_{ z < z_q}| \cdot(z - z_q)^2  \right],   
\end{equation}
where $\omega_q = \frac{2q-1}{2Q}$, $|\omega - \mathbbm{1}_{z < z_q}|$ is the weight of  regression loss penalty  and $(z - z_q)^2 $ is the square of approximation error. 
Therefore, the quantile regression loss penalizes the overestimation error with weight $1-\omega_q$ and underestimation error with weight $\omega_q$. 
The loss function (\ref{QRloss}) is derivable everywhere, 
and the problem of the return distribution modeling
becomes to minimize the quantile regression loss, i.e., 
\begin{equation}\label{1WminProj} 
\min_{ z_1, \cdots, z_Q }  \sum_{q=1}^{Q} \mathbb{E}_{Z \sim Z^{\pi}} \left[ |\omega_q - \mathbbm{1}_{ z < z_q}| \cdot(z - z_q)^2  \right].                                               
\end{equation}
Since the objective function (\ref{1WminProj}) is convex with respect to $Z_Q$,  we can find the minimizer $\{ z^{*}_q \}_{q=1,\cdots,Q}$ by using gradient-descent approaches with a computational complexity of $\mathcal{O}(Q^2)$.   
As a result, for each state-action pair, its return distribution $Z_Q(\boldsymbol{e},\Delta \Theta) $ can be approximated by a  Q-quantile $\{ z^{*}_1(\boldsymbol{e},\Delta \Theta) , \cdots, z^{*}_Q(\boldsymbol{e},\Delta \Theta)  \}$ via (\ref{1WminProj}).  

\subsection{Optimal reflection using the DR-QRL model}
In the problem of IR-aided transmissions, after observing a deviation vector $\boldsymbol{e}_l$, the IR controller can estimate the expected downlink rate for each action  $\Delta \Theta$, by computing the marginal distribution of the return distribution $Z_Q(\boldsymbol{e}_l,\Delta \Theta)$, via
\begin{align} \label{r_Z}
	 \mathbb{E} [  Z_Q(\boldsymbol{e}_l,\Delta \Theta ) ] = \frac{1}{Q} \sum_{q=1}^{Q} z_q(\boldsymbol{e}_l,\Delta \Theta).  
\end{align} 
Here, (\ref{r_Z}) evaluates the summation of future transmission rates from the next time slot $l+1$ to the end of the transmission interval, if an action $\Delta \Theta$ is applied to adjust the IR coefficient, given the current signal deviation vector $\boldsymbol{e}_l$. 
Then, the optimal reflection coefficient  ${\boldsymbol{\phi}}_{l+1} = \boldsymbol{\phi}_l \times \Delta \Theta^{*}$ that maximizes the summation of future downlink data rate is   
\begin{align} \label{equOptIRlearnig}
	\Delta \Theta^{*} = \arg\max_{\Delta \Theta} \mathbb{E} [  Z_Q(\boldsymbol{e}_l,\Delta \Theta ) ] = \arg\max_{\Delta \Theta}\frac{1}{Q} \sum_{q=1}^{Q} z_q(\boldsymbol{e}_l,\Delta \Theta) .  
\end{align} 
After the IR provides downlink service to UEs using the reflection coefficient $\boldsymbol{\phi}_{l+1}$, the new state $\boldsymbol{e}_{l+1}$ and reward $r_{l+1}$ will be updated with the IR controller at the end of the $l+1$ time slot.    
Given the downlink transmission,  the empirical distribution  can be updated via a Q-learning approach, where ${z}_i(\boldsymbol{e}_{l},\Delta \Theta_{l}) \leftarrow r_{l+1} + \gamma {z}^{l}_i(\boldsymbol{e}_{l+1},\Delta \Theta_{l+1})$, $\forall i = 1,\cdots,Q$.
In the end, the return distribution $Z^l_Q$ is updated to minimize the distance from the target distribution ${Z}^{\pi}$, based on (\ref{1WminProj}). 
The training and update algorithm of the QR-DRL model for the real-time optimization of the IR reflection is summarized in Algorithm \ref{algo3}.
 
\begin{algorithm}[!t] \small   
	\caption{QR-DRL approach to improve the IR coefficient with imperfect CSI} \label{algo3} 
	\begin{algorithmic}
		\State \textbf{Initialize}: $0 \ll \gamma < 1$  and the DRL model $Z^0_Q = \{z^0_1, \cdots, z^0_Q \}$ for each state-action pair $(\boldsymbol{e},\Delta \Theta) $. \\
		At the beginning of a coherence interval ($l=0$), the IR computes environment state $\boldsymbol{e}_0$, based on the estimated CSI, and optimizes the reflection coefficient $\boldsymbol{\phi}_0$, according to Algorithm \ref{algo1}. \\ 
		\textbf{For} $l = 1,\cdots, L-1$, the IR controller will\\
		\quad 1. Receive  UEs' feedback  $\boldsymbol{y}_l$, and compute  the deviation state $\boldsymbol{e}_l$ and the reward $r_l$; \\ 
		\quad 2. Calculate  the expected return value of the current state $\boldsymbol{e}_l$ with respect to each action $\Delta \Theta  $ by \\ 
		\quad \quad $V(\boldsymbol{e}_l,\Delta \Theta) = \frac{1}{Q}\sum_{i=1}^Q {z}^{l-1}_i (\boldsymbol{e}_l,\Delta \Theta)$; \\
		\quad 3. Adjust the reflection coefficient via  $\boldsymbol{\phi}_{l+1} = \boldsymbol{\phi}_l \times \Delta \Theta_l$, where $ \Delta \Theta_l  =\arg\max_{\Delta \Theta  } V(\boldsymbol{e}_l,\Delta \Theta)$;\\ 
		\quad 4. Update the empirical return distribution, using $\boldsymbol{e}_{l-1}$, $\Delta \Theta_{l-1}$, $\boldsymbol{e}_l$, $\Delta \Theta_{l}$, and $r_l$, via: \\ 
		\quad \quad ${z}_i(\boldsymbol{e}_{l-1},\Delta \Theta_{l-1}) \leftarrow r_l + \gamma {z}^{l-1}_i(\boldsymbol{e}_{l},\Delta \Theta_{l})$, $\forall i = 1,\cdots,Q$; \\ 
		\quad 5. Update the QR-DRL model $Z_Q^l$ by minimizing the quantile loss between ${Z}$ and $Z_Q^l$ via:\\ 
		\quad \quad   $Z^l_Q = \arg\min_{ \{z_q\}_{q=1, \cdots, Q} }$ $\sum_{q=1}^Q \sum_{i=1}^Q |\frac{2q-1}{2Q} - \mathbbm{1}_{{z}_i <z_q}|\cdot ({z}_i-z_q)^2 $; \\
		\textbf{End}  
	\end{algorithmic}
\end{algorithm}  

\subsection{Convergence of the QR-DRL method} 
Here, we analyze the convergence property of the proposed QR-DRL approach. 
For any distributions $Z_1$, $Z_2$, the maximal form of 1-Wassertein distance is defined as \cite{bellemare2017distributional}
\begin{equation} \label{maxWmetric}  
	\bar{d}_1(Z_1,Z_2) = \sup_{\boldsymbol{e},\Delta \Theta} d_1(Z_1(\boldsymbol{e},\Delta \Theta),Z_2(\boldsymbol{e},\Delta \Theta)), 
\end{equation} 
which will be used as the distance metric to establish the convergence of the QR-DRL method. 
Let $\Pi_{W_1}$ be the 1-Wasserstein minimizing quantile projection, defined in (\ref{1WminProj}), and $\mathcal{T}^\pi$ be the distributional Bellman operator \cite{bellemare2017distributional} that defines DRL iterations,  
where $\mathcal{T}^{\pi} Z(\boldsymbol{e} ,\Delta \Theta)  =  {r}(\boldsymbol{e} ,\Delta \Theta ) + \gamma Z (\boldsymbol{e}^{'},\Delta \Theta)$. 
Then, the convergence of the QR-DRL projection $\Pi_{W_1}\mathcal{T}^\pi$, combined by quantile regression with the DRL operator, to a unique fixed point  is given as follow. 
\begin{theorem}   \label{theorem1}
	The distance between the Q-quantile approximation $Z_Q$ and the target return distribution $Z$, in terms of maximal 1-Wassertein metric, will converge to zero, via the repeated application of the QR-DRL projection $\Pi_{W_1}\mathcal{T}^\pi$, i.e. 
	$	\bar{d}_1(\Pi_{W_1}\mathcal{T}^\pi Z, \Pi_{W_1}\mathcal{T}^\pi Z_Q) \leq \gamma (\bar{d}_1(  Z,  Z_Q) + C_Q )$, 
	where $\gamma \in (0,1)$ is the discount factor,  $C_Q$ is a finite number that depends on the value of $Q$, and $\lim_{Q\rightarrow \infty} C_Q =0$.  
\end{theorem}
\begin{proof}
	See Appendix A. 
\end{proof}
Theorem \ref{theorem1} shows that by applying the QR-DRL operator $\Pi_{W_1}\mathcal{T}^\pi$ once, the distance between $Z$ and $Z_Q$ will be changed from $\bar{d}_1(  Z,  Z_Q)$ to be lower than $\gamma( \bar{d}_1(Z,  Z_Q) +   C_Q )$.  
Given that $\bar{d}_1(Z,  Z_Q)$ and $ C_Q$ both have finite values and $\gamma \in (0,1)$, after the repeated application of the projection $\Pi_{W_1}\mathcal{T}^\pi$, the distance between  $Z$ and $Z_Q$  will eventually go to zero. 
Therefore, the iterative approach based on the QR-DRL projection  $\Pi_{W_1}\mathcal{T}^\pi$ will converge the approximation $Z_Q$ to the unique fixed point $Z$, following the policy $\pi$ in (\ref{max_return}).   

Consequently, we conclude that the proposed learning approach in Algorithm \ref{algo3} will converge to a unique and stable distribution of the downlink sum-rate for the IR-aided transmission, as defined in (\ref{return}), and, as a result, the optimal policy $\pi^*$ that optimizes the IR coefficient $\boldsymbol{\phi}$  can be uniquely determined, based on (\ref{r_Z}) and (\ref{equOptIRlearnig}), which solves the optimal problem in (\ref{equOpt_b}).  
Therefore,  the reflection optimization with limited  CSI can be solved, based on to Algorithm 2 and 3, at the end of each time slot, so that the sum of future data rate  can be maximized. 
Meanwhile, as shown in Algorithm 3,  the scalar value of $z$ is updated using a Q-learning algorithm in step 4. After that, our proposed  DRL approach applies one more step  to model the distribution function of the random return value using the quantile optimization  in step 5.  
Thus, compared with  the scalar-based $z$ in the Q-learning RL method, the proposed DRL approach captures the frequency and value of the corresponding transmission rate  with respect to each chosen reflection coefficient. 
Therefore, due to the return distribution model,  the proposed DRL method can find the expected reward of the sum-rate more accurately, compared to a conventional scalar-based  RL  model.  

\section{Simulation Results and Analysis}\label{sec_simulation}

\subsection{Setting}  
\begin{table}[t!]  \vspace{-0.2cm}
	\centering
	\caption{Simulation and algorithm parameters}
	\label{table1}
	\begin{tabular}{ |p{5.5cm}|p{2cm}||p{5.5cm}|p{2cm}|  }
		\hline
		Parameters & Values & 	Parameters & Values \\
		\hline
		Bandwidth $b$  & $2$ MHz   & Number of UEs $K$ & 4 \\
		Noise power spectrum density at UE $\sigma^2$  & $-174$ dBm/Hz    & Noise power spectrum density at BS $\sigma_{BS}^2$  & $-170$ dBm/Hz\\
		Transmit power at BS $P_\text{max}$ &  $40$ dBm  & 	Transmit power at UE $p_c$ &  $10$ dBm\\
		Time slot $\tau$& $0.1$ s & Coherence   interval $T$ & 10$\tau$ \\
		Uplink training subphase $\tau_c$  \cite{nadeem2019intelligent} & $0.01\tau$ & Data transmit phase $\tau_d$ & $(1- 0.01N)\tau$\\
		Number of quantiles $Q$ & $40$  & Discount factors $\gamma$  & $0.9$ \\
		\hline
	\end{tabular}\vspace{-0.2cm}
\end{table}

In the simulations, we consider a uniform square array of antennas at the BS with $M=16$, and a uniform square array of IR with $N=16$.  
The location of BS is fixed at $(0,0,25)$, 
the location of IR is $(0, 20, 30)$
and the location of each hotspot UE is modeled by an i.i.d. two-dimensional Gaussian $(\boldsymbol{x}_k,\boldsymbol{y}_k) \sim \mathcal{N}((0,20),5\boldsymbol{I}_2)$ with zero height.  
The path loss in dB is modeled \cite{5GCM} as $PL[dB] = 32.4  + 21 \log_{10}(d) + 20 \log_{10}(f) + \xi$, where $d$ denotes the distance between the transmitter and receiver in meter, $f$ is that the carrier frequency, and $\xi\sim \mathcal{CN}(0,\sigma^2_{sf})$ denotes the small fading parameter. 
Given the carrier frequency $f=30$ GHz, we set $\sigma_{sf} = 3.762$ for a LOS link, and $\sigma_{sf} = 8.092$ for a NLOS link. 
The MIMO channel between the BS and the IR is modeled by a rank-one matrix \cite{el2012capacity} 
	$\boldsymbol{H} =   \boldsymbol{a}^r(\varphi^r_0,\psi^r_0) \boldsymbol{a}^t(\varphi^t_0,\psi^t_0)^H$,
where $\varphi^r_0$,  $\psi^r_0$  are the azimuth and elevation  angles of arrival at IR, and  $\varphi^t_0$,   $\psi^t_0$  are the azimuth and elevation  angles of departure at  the BS.  
The vectors $\boldsymbol{a}^r(\varphi^r_0,\psi^r_0)$ and  $\boldsymbol{a}^t(\varphi^t_0,\psi^t_0)$ are the normalized receive/transmit array response
vectors at the corresponding angles of arrival/departure, where 
\begin{equation*}
\begin{aligned}
	\boldsymbol{a}(\varphi,\psi) = [1, \cdots, & \exp( j  \pi( v \sin(\varphi) \sin(\psi) + w\cos(\psi) )), \cdots, \\
	& \exp( j  \pi( (V-1) \sin(\varphi ) \sin(\psi) +(W-1)\cos(\psi) )) ]^T.
\end{aligned}
\end{equation*} 
For the receiver IR, $\boldsymbol{a}^r(\varphi^r_0,\psi^r_0) = \boldsymbol{a}(\varphi^r_0,\psi^r_0)$ with $V^r = W^r = \sqrt{N}$, and for the transmitter BS,  $\boldsymbol{a}^t(\varphi^t_0,\psi^t_0) = \boldsymbol{a}(\varphi^t_0,\psi^t_0)$ with $V^t = W^t = \sqrt{M}$. 
Similarly, the MISO channel from the IR to each UE $k$ is modeled by a vector $\boldsymbol{h}_k(\varphi_k,\psi_k) = \boldsymbol{a}(\varphi_k,\psi_k)$ with $V_k=W_k=\sqrt{N}$, where $\varphi_k,\psi_k$ are the azimuth and elevation angles of departure at  the IR towards UE $k$.  
Thus, for each UE $k \in \mathcal{K}$, the MISO downlink  channel is given by $\boldsymbol{G}_k = \frac{\exp(-j \frac{2\pi}{\lambda}d_{k})}{\sqrt{PL_k}} \text{diag}(\boldsymbol{h}_k) \boldsymbol{H} $, where $d_k$ and $PL_k$ are the distance and  path lossof the connected BS-IR-UE link, respectively.   
The values of other parameters are given in Table \ref{table1}.

\subsection{IR-aided transmission with perfect CSI}
 
In this section, we evaluate the performance of the IR-aided communication, given perfect downlink CSI.  
In order to show the advantage of the proposed  approach,  two baseline schemes are introduced, which are the direct transmission without IR reflection  and the fixed IR reflection.  
Note that, the optimal precoding at the BS is applied in three schemes. 
However, for the direct transmission, the BS optimizes the precoding matrix, based on the direct downlink channel. 
Meanwhile, after the optimization at the BS side,  the fixed reflection method keeps the IR coefficient $\boldsymbol{\phi} = \boldsymbol{I}_N$, while the proposed  approach optimizes the reflection coefficient $\boldsymbol{\phi}^{*}$, based on to Algorithm 1.   
The metric that is used to evaluate the communication performance is the time-average downlink sum-rate within each coherence interval $T$. 
For the proposed approach, the time-average sum-rate is calculated via $(1-\frac{N\tau_c + \tau_m}{T})r^{\text{}}$, where $N\tau_c$ is the channel estimation time, and $\tau_m$ is for signal processing.  
However, for the direct transmission and  the fixed IR reflection, the channel measurement is only one $\tau_c$, and, thus, the time-average sum-rate is  $(1-\frac{\tau_c + \tau_m}{T})r$.  
 
\subsubsection{Performance evaluation}

\begin{figure*}[!t]
	\begin{center}
		\vspace{-1cm}
		\begin{subfigure}{.49\textwidth}
			\includegraphics[width=8.9cm]{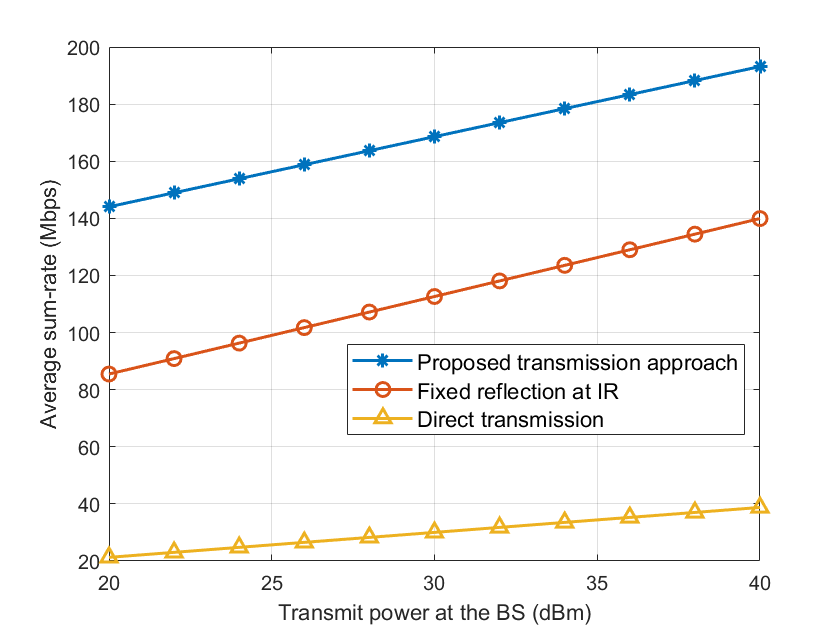}
			\caption{\label{aveRate_propose_fixedIR_direct} The time-average downlink sum-rate increases, as the transmit power at the BS increases.  
			}
		\end{subfigure}
		\begin{subfigure}{.49\textwidth}
			\centering
			\includegraphics[width=8.9cm]{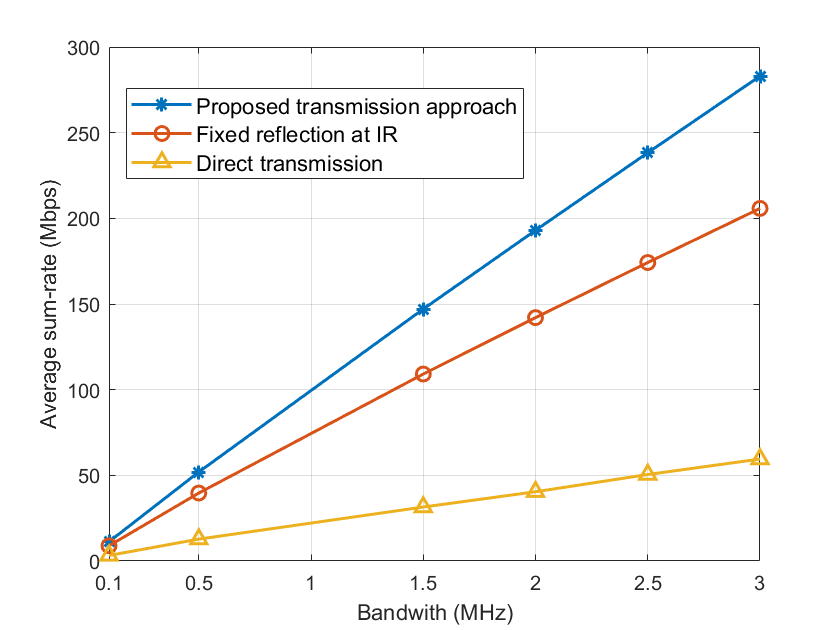}\vspace{0 cm}
			\caption{\label{aveRate_bandwidth} The time-average sum-rate increases, as the downlink bandwidth  increases.  
			}
		\end{subfigure}
		\vspace{-0.2 cm} 
		\caption{\small{\label{performance} The proposed approach outperforms two baseline schemes, as the transmit power and  bandwidth increase.}  
		}
	\end{center}
	\vspace{- 1.2  cm}
\end{figure*}  
 
Fig. \ref{aveRate_propose_fixedIR_direct} shows that, as the BS transmit power increases from $20$ to $40$ dBm, the time averages of the downlink sum-rate increases in all three schemes, and the proposed  method outperforms both baseline schemes. 
First, compared with the direct transmission without IR, the fixed reflection scheme yields  $3$-fold increase in downlink transmission performance, due to better channel state.   
Since the direct link from the BS to each UE is usually NLOS, the high path loss yields a small received power at each UE, and thus, the average downlink rate of the direct transmission is always lower than $40$ Mbps.   
However, via a passive reflection of the IR even with a fixed coefficient, the  BS-IR-UE channels can be established with LOS links, and, hence, the average rate increases significantly.   
Furthermore, compared with the fixed reflection scheme, our proposed  method further improves the reflection coefficient thus achieving a performance gain of over $50$ Mbps. 
As the transmit power of the BS increases from $20$ to $40$ dBm, the proposed approach yields a performance gain of  over $30\%$,  compared with the fixed reflection scheme.  

Fig. \ref{aveRate_bandwidth} shows that as the downlink bandwidth increases,  the average downlink sum-rate of all three methods increase, and the proposed approach outperforms both baselines.  
Compared with the fixed reflection scheme, the downlink sum-rate of the proposed approach yields a performance gain from $31.63\%$ to $43.20\%$, as the bandwidth increases from $0.1$ to $3$ MHz. 
Meanwhile, compared with the direct transmission,   the proposed approach shows $2$-fold increases in the downlink sum-rate. 
From  Fig. \ref{performance}, we can conclude that a large bandwidth and a higher transmit power both yield a larger downlink sum-rate for the IR-aided transmission.

\subsubsection{Impact of the number of antennas}

Fig. \ref{aveRate_vs_numAnte}  shows the relationship between the downlink sum-rate  and  the number of antennas at the BS and IR,  respectively. 
First,  Fig. \ref{numBSAnte} shows that, when the IR has a fixed number of reflective elements $N=16$, by increasing the number of BS antennas $M$  from $16$ to $100$, the average sum-rate increases for all three methods, and the proposed approach yields the best performance. 
For a larger number of transmit antennas, the diversity in fading multi-path channels can be achieved to improve link reliability and facilitate  the beamforming precoding.     
Therefore, the average downlink rate increases when we have more transmit antennas.   
Fig. \ref{numIRcomp} shows that, when the BS has a fixed number of antennas $M=16$, as the number of IR reflection components increases from $16$ to $100$, the average sum-rate of the proposed method  and the fixed reflection scheme will both increase,  while the performance of direct transmission remains the same.   
By deploying more reflective components at the IR, more LOS paths are created between the BS and downlink UEs, such that the received signal power at UEs will increase thus improving downlink transmission rate. 
Since the direct BS-UE link does not pass through the IR, the number of IR components does not impact the direct transmission scheme.  
Furthermore,  by comparing Figs. \ref{numBSAnte} and \ref{numIRcomp}, we can see that deploying more IR elements is more efficient to increase the downlink transmission rate, compared with increasing the number of  BS antennas. 
Even though the downlink rate can be significantly improved by deploying more IR elements,   as the number of IR components becomes larger,  more time is required for channel estimation, and the effective transmission time during each coherence interval will decrease. 
Therefore, when the number of IR elements is $N>80$, the increasing slope of the IR-aided transmission rate becomes much smaller.

\begin{figure*}[!t]
	\begin{center}
		\vspace{- 1 cm}
		\begin{subfigure}{.49\textwidth}
			\centering
			\includegraphics[width=8.9cm]{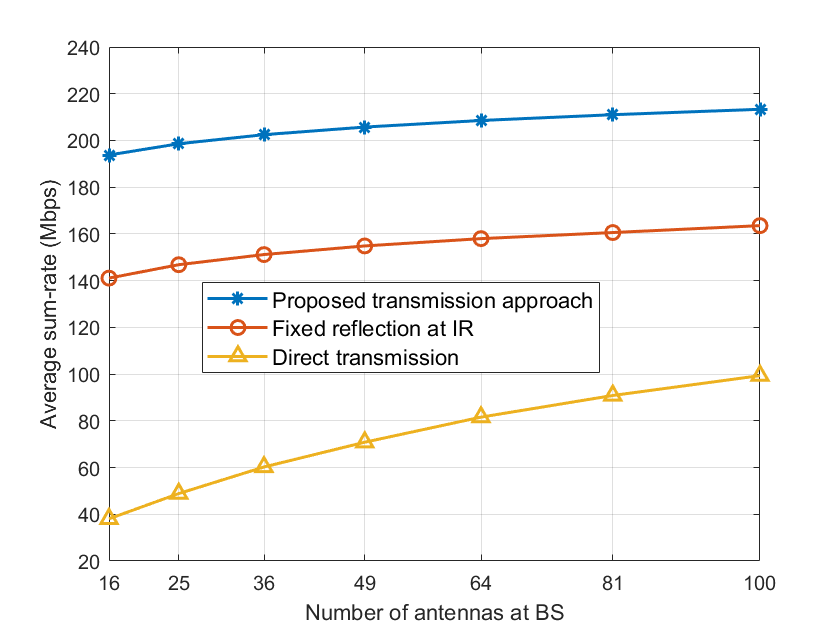} 
			\caption{\label{numBSAnte} 
			The time-average downlink sum-rate increases, as the number $M$ of BS antennas increases, for $N = 16$.   
			}
		\end{subfigure}
		\begin{subfigure}{.49\textwidth}
			\centering
			\includegraphics[width=8.9cm]{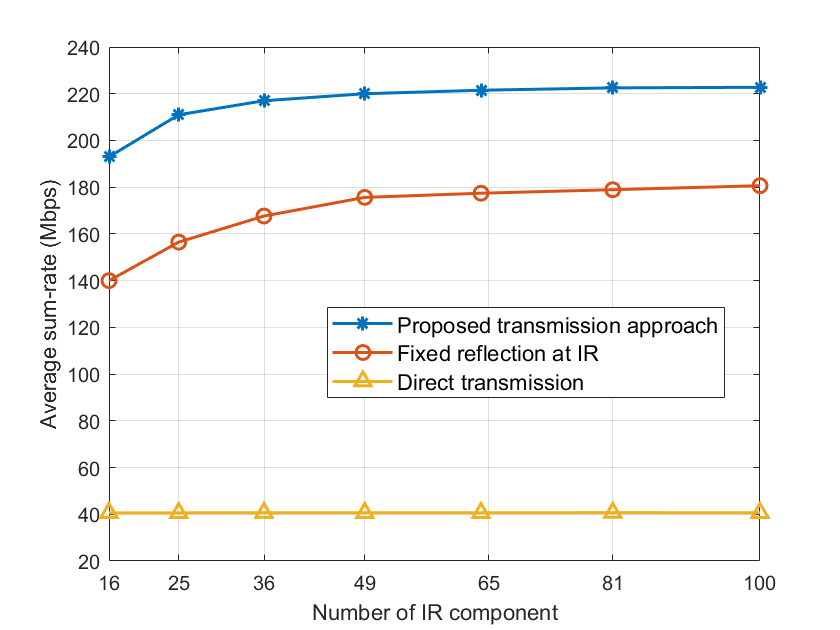} 
			\caption{\label{numIRcomp} 
			The time-average downlink sum-rate increases, as the number $N$ of IR components increases, for $M = 16$. 
			}
		\end{subfigure}
		\vspace{-0.2 cm} 
		\caption{\small{\label{aveRate_vs_numAnte} The time-average downlink sum-rate increases, as the number of BS antennas or IR components increases.}  
		}
	\end{center}
	\vspace{-1.2 cm}
\end{figure*}

\subsection{IR-aided transmission with imperfect CSI}

In this section, we evaluate the performance of the IR-aided downlink transmission, given limited knowledge of CSI.   
Algorithm 1 is used to initialize the transmission system at the beginning of each communication interval $T$, and then, the proposed QR-DRL method will be applied at the end of each time slot $\tau$, such that the distribution of downlink rate can be learned, based on the UEs' feedback, in order to improve the reflection coefficient of the IR.   

\subsubsection{Preprocessing}

The DRL approach aims to model a return distribution for each state-action pair $(\boldsymbol{e},\Delta \Theta)$. 
Given that the state-action space is continuous, it is necessary to have a discrete state-action space, such that the number of state-action pairs, as well as the estimated distribution functions, is finitely countable.  
First,  the state space  is reduced to be a $K$-dimensional binary space, where for each $k = 1,\cdots,K$, $e_k = 0$, if $|y_{k} -\hat{y}_{k}|^2 \le E_{th}$; otherwise, $e_k = 1$. Here, $E_{th}$ is a threshold of the highest acceptable power for the signal deviation $e$.  
Therefore,  $e_k = 0$ means that the signal deviation is small, and the downlink channel error is acceptable; 
otherwise, $e_k = 1$ indicates that 
the measured CSI $\hat{\boldsymbol{G}}_k$ of UE $k$ is  significantly different from  the actual downlink CSI $\boldsymbol{G}_k$. 
Consequently, the number of possible states  is $2^K$.   
Second, in the action space, there are two possible phase change for each of the IR element $n$, where  $\Delta \theta_n = \exp(j \pi) = -1$ means to shift the phase by $\pi$, and $\Delta \theta_n =\exp(j 0  ) = 1$ means to keep the phase of the $n$-th IR component the same. 
Therefore, the number of possible actions  is $2^N$, 
which increases exponentially with the number of IR components.  
In order to reduce the size of action space, an exhausting search algorithm is first applied to find a subsection of the action space. 
As shown in Fig. \ref{actionCDF}, for a fixed number of IR components $N=16$, 
we can reduce the number of actions from $2^{16}$ to $60$, which guarantees to cover the optimal action set with a probability of around $99\%$. 

\begin{figure*}[!t]
	\begin{center}\vspace{-1  cm}
		\includegraphics[width=8.7cm]{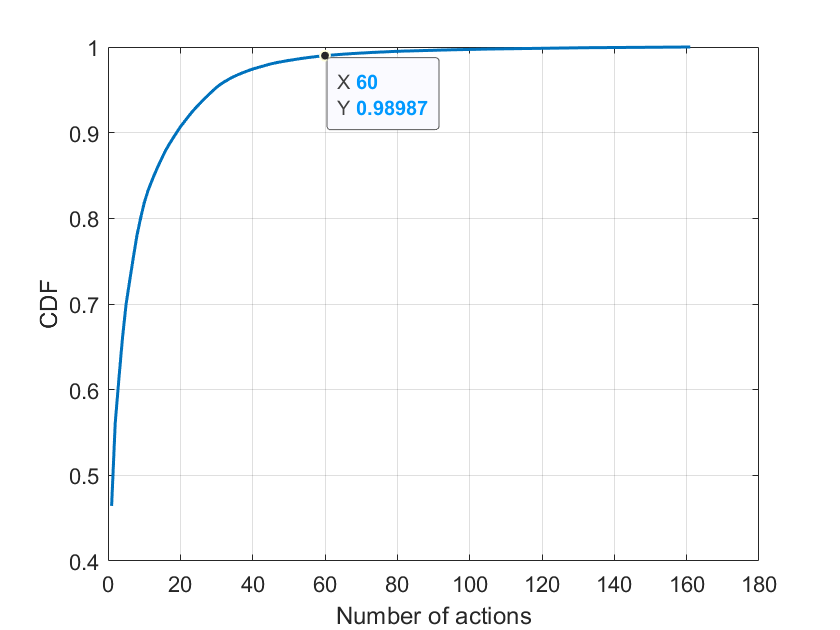}
		\caption{\label{actionCDF}  A reduced set with the top $60$ of the most frequent actions can cover the optimal result with an empirical probability of $99\%$.  
		}
	\end{center}
	\vspace{-1  cm}
\end{figure*}

\subsubsection{Training process}

Similar to most RL-based algorithms, the proposed QR-DRL method can be slow to converge. 
In order to enable an efficient reflection performance, it is necessary to  train the distribution model, before the online deployment.
Here, a traditional RL scheme, based on  Q-learning algorithm, is introduced to compare the performance of the proposed method\footnote{Due to the use of neural network, deep Q-learning has a much higher computational complexity, compared with the proposed  QR-DRL approach. The comparison between deep Q-learning and deep DRL methods will be subject to our future work.}.    
Meanwhile, the result from Algorithm \ref{algo1} is used as the baseline, 
and 
the optimal result, obtained based on perfect CSI, is used as the training target.

Fig. \ref{trainingPreocess}  illustrates the training process, in which the downlink sum-rate of the IR-aided transmission is averaged per $300$ simulation episodes. 
First, Fig. \ref{trainingPreocess} shows that the proposed QR-DRL approach converges, which supports the proof in Theorem \ref{theorem1}. 
Second, compared with the Q-learning scheme, the proposed QR-DRL method has a larger variance and converges more slowly. 
Here, we note that the gap between the proposed learning method and the optimal result is caused by the discretion of the state-action space. 

 \begin{figure*}[!t]
 	\begin{center}
 		\vspace{- 1 cm}
 		\begin{subfigure}{.49\textwidth}
 			\centering
 			\includegraphics[width=8.7cm]{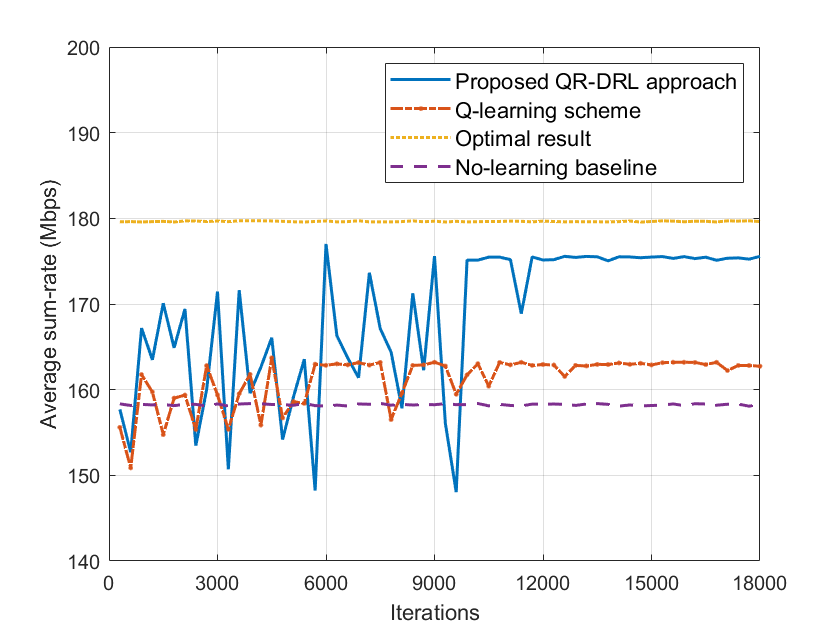} 
 			\caption{\label{trainingPreocess} Training process: Average sum-rate per 300 episodes. 
 			}
 		\end{subfigure}
 		\begin{subfigure}{.49\textwidth}
 			\centering
 			\includegraphics[width=8.7cm]{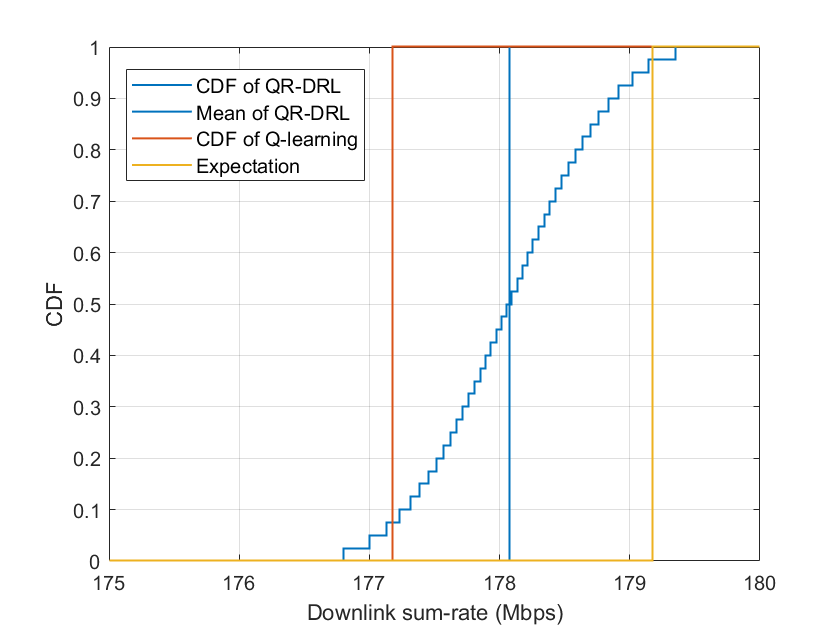} 
 			\caption{\label{state16_action13108_QR_QL} Trained distribution for  state-action  $( \boldsymbol{1}_K, \Delta \Theta_{13108})$.
 			}
 		\end{subfigure}
 		\vspace{-0.2 cm} 
 		\caption{\small{\label{proprecess_train}  Average sum-rate during the training process (left) and the training result of the state-action pair $(\boldsymbol{1}_K,\Delta \Theta_{13108})$ (right).
 			} 
 		}
 	\end{center}
 	\vspace{-1.2  cm}
 \end{figure*}

In order to provide more details on the training result, we focus on a worst-case scenario, where the deviation state $\boldsymbol{e}=\boldsymbol{1}_K$. 
In this state, the received signal at each UE is significantly different from the expectation, and, therefore, the measured CSI has a large error. 
We compare the proposed QR-DRL approach and traditional Q-learning scheme to optimize the reflection parameter under this worst state.  
Fig. \ref{state16_action13108_QR_QL} shows the CDF of the trained return distribution for the state $\boldsymbol{e}= \boldsymbol{1}_K$ with its optimal action No. 13108\footnote{To convert an action $\Delta \Theta$ to its action number, first, we replace $-1$ by $1$ and, then, we replace $1$ by $0$ on the diagonal of the action matrix. Second, we convert the diagonal binary vector to a decimal number and we add one. Thus, $\Delta \Theta_{13108}=\text{diag}(1, 1, -1, -1, 1, 1, -1, -1, 1, 1, -1, -1, 1, 1, -1, -1)$.}.  
Fig. \ref{state16_action13108_QR_QL} shows that compared with  Q-learning, the mean of the trained distribution from the proposed method is closer to the real expectation, which is calculated based on the error-free CSI.   
Instead of choosing the optimal action  $\Delta \Theta_{13108}$, Q-learning is more likely  to select suboptimal actions  to adjust the reflection parameters, due to its larger prediction error. 
Different from Q-learning that only learn the expected returns, the QR-DRL method models the return as a distribution function for each state-action pair. 
Compared with a deterministic expectation value, a distribution function can capture more details of the uncertainty of downlink rate, and the QR-DRL method yields a more accurate prediction than Q-learning, for the specific state-action $(\boldsymbol{1}_K,\Delta \Theta_{13108})$.



\subsubsection{Performance evaluation of online deployment}

Fig. \ref{onine} shows the online performance of the proposed QR-DRL approach, with the no-learning baseline and  Q-learning scheme,  given imperfect downlink CSI.  We run each method for $100$ times and plot the average results. 
After an offline training stage, both the proposed QR-DRL and Q-learning models have converged. Therefore, the downlink sum-rate has a smaller variation. compared with the training phase in Fig. \ref{trainingPreocess}. 
Meanwhile, Fig. \ref{onine} shows that the proposed  QR-DRL approach improves the average data rate of the IR-aided downlink communication by around $10\%$, compared with the conventional RL method. 
Therefore, we conclude that the distribution-based prediction of the downlink sum-rate for the future time slot is more accurate, compared with a scalar-based prediction of Q-learning, and the proposed QR-DRL method enables an efficient communication  of the IR-aided mmW transmission service. 

 \begin{figure*}[!t]
 	\begin{center}\vspace{-1 cm}
 		\includegraphics[width=8.7cm]{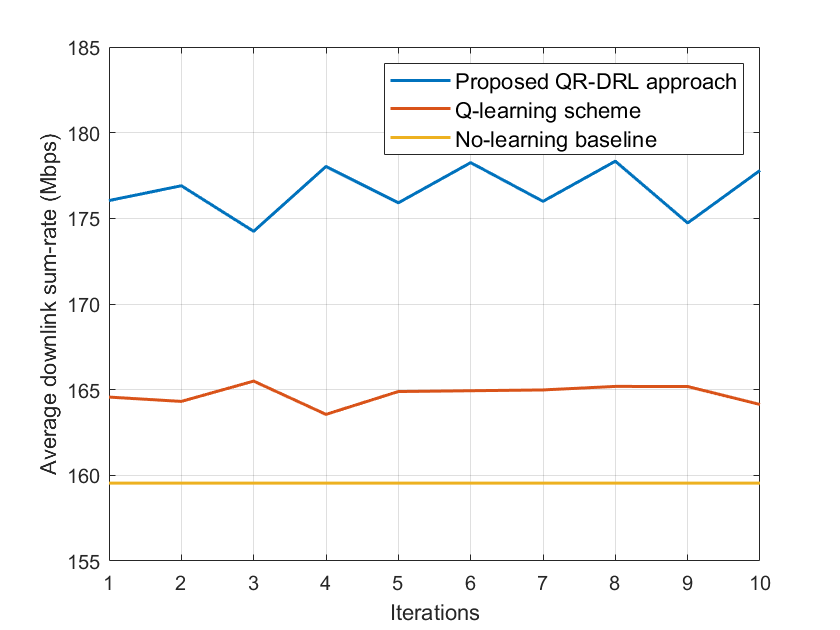}
 		\caption{\label{onine}  Transmission performance for the online deployment.
 		}
 	\end{center}
 	\vspace{-1.2 cm}
 \end{figure*}

\section{Conclusion}\label{sec_conclusion}
In this paper, we have proposed a novel framework to optimize the downlink multi-user communications of a mmW BS that is assisted by an IR. 
We have developed a practical approach to measure the real-time CSI.  
First, for a perfect CSI scenario, the precoding transmission of the BS and the reflection coefficient of the IR are jointly optimized, via an iterative approach, so as to maximize the sum of downlink rates towards multiple users.   
Next, given imperfect CSI, we have proposed a DRL approach to learn the optimal IR reflection, so as to maximize the expectation of downlink sum-rate. 
In order to model the rate's probability distribution, we have developed an iterative algorithm and proved the convergence of the proposed QR-DRL approach.  
Simulation results show that, given  error-free CSI, the proposed transmission approach outperforms two baselines: the direct transmission scheme and a fixed IR reflection.    
Furthermore, under limited knowledge of CSI, simulation results show that the proposed QR-DRL method improves the average data rate by around $10\%$ in online deployments, compared with a Q-learning baseline. 
In our future work, we will investigate scenarios related to the IR-aided communication for multi-antenna UEs with multipath propagation.

\begin{appendices}
	\section{Proof of Theorem \ref{theorem1}}   
	In order to prove that 	
	for any distribution $Z_1$, $Z_2 \in \mathcal{Z}$, $\bar{d}_1(\Pi_{W_1}\mathcal{T}^\pi Z_1, \Pi_{W_1}\mathcal{T}^\pi Z_2) \leq \gamma (\bar{d}_1(  Z_1,  Z_2) + C_Q )$,  
	we only need to prove that,   
	for  $U = \mathcal{T}^\pi Z_1$ and $V = \mathcal{T}^\pi Z_2$, 
	\begin{equation}\label{appEqu2}
	\bar{d}_1(\Pi_{W_1} U, \Pi_{W_1} V) \leq \bar{d}_1(U, V) + C_Q, 
	\end{equation} 
	where $\lim_{Q\rightarrow \infty} C_Q =0$. Then, based on Lemma 3 in \cite{bellemare2017distributional}, we have
	\begin{equation}\label{appEqu1}
		\bar{d}_1(  \mathcal{T}^\pi Z_1, \mathcal{T}^\pi Z_2) \leq \gamma \bar{d}_1 (  Z_1,  Z_2).
	\end{equation}    
	Combining (\ref{appEqu2}) and (\ref{appEqu1}), we have 
	$\bar{d}_1\left[ \Pi_{W_1} (\mathcal{T}^\pi Z_1), \Pi_{W_1} (\mathcal{T}^\pi Z_2)  \right] \leq 
	 \bar{d}_1(\mathcal{T}^\pi Z_1, \mathcal{T}^\pi Z_2) + C_Q \leq 
	 \gamma (\bar{d}_1(  Z_1,  Z_2)+C_Q)$,
	which completes the proof. Next, we will focus on the proof of (\ref{appEqu2}). 

	For notional simplification, we denote the cumulative probability function of any distribution $Z \in \mathcal{Z}$ as $F_Z(\theta) = \omega$, and the  cumulative probability function of a Q-quantile estimation of $Z$ as $F_{Z_Q}(\theta_i) = \frac{i}{Q}$, $i=1,\cdots, Q$. 
	The $p$-Wasserstein distance between two distributions $U$ and $V$ is defined, based on (\ref{Wmetric}), by 
	\begin{equation}\label{Wuv}
		{d_p}^p (U,V) =  \int_{0}^{1} \left|F^{-1}_{U}(\omega) - F^{-1}_{V}(\omega)\right|^p \dif \omega.  
	\end{equation} 	  
	Then, the maximal form of $p$-Wasserstein metric is 
	${\bar{d}_p}^{~p} ( U,  V) = \sup_{\boldsymbol{y},\boldsymbol{\phi}} {d_p}^{p}(  U(\boldsymbol{y},\boldsymbol{\phi}), V(\boldsymbol{y},\boldsymbol{\phi}))$. 
	Without loss of generality, we assume that  ${\bar{d}_p}^{~p} ( U,  V) $ is finite.
	After the $1$-Wasserstein minimizing projection, the $p$-Wasserstein distance between Q-quantile estimations $\Pi_{W_1} U$ and $\Pi_{W_1} V$ is 
	 \begin{equation}\label{WQuv}
	 	\begin{aligned} 
	 		{d_p}^{p} (\Pi_{W_1} U, \Pi_{W_1} V)  &=  \int_{0}^{1} \left|{F_{U_Q}}^{-1}(\omega) - {F_{V_Q}}^{-1}(\omega)\right|^p \dif \omega 
	 		 & = \frac{1}{Q} \sum_{i=1}^{Q} \left|{F_{U_Q}}^{-1}({\omega}_i) - {F_{V_Q}}^{-1}({\omega}_i) \right| ^p , 
	 	\end{aligned}
	 \end{equation}   
	 where ${\omega}_i = \frac{2i-1}{2Q}$. Based on (\ref{Wuv}) and (\ref{WQuv}), we have for $p=1$,
	 \begin{equation*}
	 {d_1} (\Pi_{W_1} U, \Pi_{W_1} V) -  {d_1} (U,V)   
	  = \int_{0}^{1} \left|{F_{U_Q}}^{-1}(\omega) - {F_{V_Q}}^{-1}(\omega)\right| - \left|F^{-1}_{U}(\omega) - F^{-1}_{V}(\omega)\right| \dif \omega
	  \end{equation*}
	 \begin{equation*}
	 	 \stackrel{a}{\leq} \int_{0}^{1} \left|{F_{U_Q}}^{-1}(\omega) - {F_{V_Q}}^{-1}(\omega)   - (F^{-1}_{U}(\omega) - F^{-1}_{V}(\omega))\right|  \dif \omega
	 \end{equation*}
	 \begin{equation*}
	 = \int_{0}^{1} \left|({F_{U_Q}}^{-1}(\omega)  -  F^{-1}_{U}(\omega)) + (F^{-1}_{V}(\omega) - {F_{V_Q}}^{-1}(\omega) )   \right|  \dif \omega
	 \end{equation*}
	 \begin{equation} \label{Winequ}
	 \stackrel{b}{\leq}  \int_{0}^{1} \left| {F_{U_Q}}^{-1}(\omega)  -  F^{-1}_{U}(\omega) \right|  \dif \omega  + \int_{0}^{1} \left|F^{-1}_{V}(\omega) - {F_{V_Q}}^{-1}(\omega)    \right|  \dif \omega = \mathcal{L}_U(U_Q) + \mathcal{L}_V(V_Q),
	 \end{equation} 
	 where $\mathcal{L}_Z(Z_Q)$ is the estimation error between  distribution $Z$ and the Q-quantile estimation $Z_Q$. 
	 In step (a), we have used the triangle inequality
	 of  absolute values: $|A| - |B| \le |A - B|$, where $A = {F_{U_Q}}^{-1}(\omega) - {F_{V_Q}}^{-1}(\omega)$ and $B = F^{-1}_{U}(\omega) - F^{-1}_{V}(\omega)$, and step (b) follows from $|C+D| \leq |C| + |D|$, where $C = {F_{U_Q}}^{-1}(\omega)  -  F^{-1}_{U}(\omega)$ and $D = F^{-1}_{V}(\omega) - {F_{V_Q}}^{-1}(\omega)$. 
	 
	 The Q-quantile estimation error has a upper limit $\sup_Z \mathcal{L}_Z(Z_Q)= \frac{\bar{\theta}}{2Q}$, where $\bar{\theta}$ is the maximal value of $\theta$. 
     Since $\bar{d}_1(U,V)< \infty$, $\bar{\theta}$ is a finite value. 
	 Therefore, based on (\ref{Winequ}), we have 
	 \begin{equation*}
	 \sup_{\boldsymbol{y},\boldsymbol{\phi}} {d_1} \left( (\Pi_{W_1} U(\boldsymbol{y},\boldsymbol{\phi}), \Pi_{W_1} V(\boldsymbol{y},\boldsymbol{\phi}))  \right) \leq \sup_{\boldsymbol{y},\boldsymbol{\phi}} \left(   {d_1}  (U(\boldsymbol{y},\boldsymbol{\phi}),V(\boldsymbol{y},\boldsymbol{\phi}))  +  \mathcal{L}_U(U_Q) +  \mathcal{L}_V(V_Q)  \right), 
	 \end{equation*} 
	 i.e., 
	 \begin{equation}\label{WinequSup}
	 {\bar{d}_1} (\Pi_{W_1} U, \Pi_{W_1} V) \leq {\bar{d}_1} (U,V)   + \frac{\bar{\theta}}{Q}.
	 \end{equation} 
	 By combining (\ref{WinequSup}) and (\ref{appEqu1}), we have 
	 $\bar{d}_1\left[ \Pi_{W_1} (\mathcal{T}^\pi Z_1), \Pi_{W_1} (\mathcal{T}^\pi Z_2)  \right] \leq 
	 \bar{d}_1(\mathcal{T}^\pi Z_1, \mathcal{T}^\pi Z_2) +  \frac{\bar{\theta}}{Q} \leq 
	 \gamma (\bar{d}_1(  Z_1,  Z_2) +  \frac{\bar{\theta}}{Q} )$.  
	 Therefore, $C_Q = \frac{\bar{\theta}}{Q}$, and  $\lim_{Q\rightarrow \infty} C_Q = \lim_{Q\rightarrow \infty}  \frac{\bar{\theta}}{Q} = 0$. 
	Consequently, let $Z_1 = Z$ and $Z_2 = Z_Q$, we have  
	$\bar{d}_1\left[ \Pi_{W_1} (\mathcal{T}^\pi Z), \Pi_{W_1} (\mathcal{T}^\pi Z_Q)  \right]  \leq 
	\gamma (\bar{d}_1(  Z,  Z_Q) + C_Q )$.   
	
	 Given that  both $\bar{d}_1(  Z,  Z_Q)$ and $C_Q$ have finite values and $\gamma \in (0,1)$, a repeated application of the QR-DRL projection  $\Pi_{W_1} \mathcal{T}^\pi$ will contract the distance between $Z$ and $Z_Q$ to zero. 
	 Therefore, the quantile approximation $Z_Q$ will eventually converge to a unique fixed point $Z$.

\end{appendices}



 
\bibliographystyle{IEEEtran}
\bibliography{references}

\end{document}